\documentclass[11pt]{article}

\usepackage{amsfonts,amsmath,amssymb, bbm, bm, xspace}
\usepackage{multirow,graphicx, epsfig,subfigure,algorithm,algorithmic,rotating}
\usepackage{cite,url}
\usepackage{fullpage}
\usepackage[small,bf]{caption}
\usepackage[table]{xcolor}
\usepackage{tabularx}
\usepackage{psfrag}
\usepackage{subfigure}
\usepackage{verbatim}
\newtheorem{theorem}{Theorem}
\newtheorem{proposition}{Proposition}

\newtheorem{lemma}{Lemma}
\newtheorem{corollary}{Corollary}
\newtheorem{definition}{Definition}

\def \endprf{\hfill {\vrule height6pt width6pt depth0pt}\medskip}
\newenvironment{proof}{\noindent {\bf Proof} }{\endprf\par}

\numberwithin{equation}{section}

\newcommand{\argmin}{\operatorname{argmin}}

\newcommand{\Pa}{\mathbb{P}}
\newcommand{\code}{\mathcal{C}}
\newcommand{\Ne}{\mathcal{N}}
\newcommand{\PP}{\mathbb{PP}} 
\newcommand{\Proj}{\Pi}

\newcommand{\Lag}{\mathcal{L}}

\newcommand{\eq}[1]{(\ref{eq:#1})}

\newcommand{\mbf}{\boldsymbol}

\newcommand{\main}{\text{\textit{\textbf{m}}}_{\rightarrow j}}
\newcommand{\maout}{\text{\textit{\textbf{m}}}_{j \rightarrow}}
\newcommand{\mia}{m_{i\rightarrow j}}

\newcommand{\mai}{m_{j\rightarrow i}}
\newcommand{\lai}{\lambda_{j,i}}

\newcommand{\bfH}{\mbf{H}}
\newcommand{\bfP}{\mbf{P}}
\newcommand{\bfS}{\mbf{S}}
\newcommand{\bfQ}{\mbf{Q}}

\newcommand{\bfzero}{\mathbf{0}}

\newcommand{\bfnu}{\mbf{\nu}}

\newcommand{\bfgamma}{\mbf{\gamma}}
\newcommand{\bflambda}{\mbf{\lambda}}
\newcommand{\bftheta}{\mbf{\theta}}
\newcommand{\bfeta}{\mbf{\eta}}

\newcommand{\bfa}{\mbf{a}}
\newcommand{\bfb}{\mbf{b}}
\newcommand{\bfe}{\mbf{e}}
\newcommand{\bff}{\mbf{f}}
\newcommand{\bft}{\mbf{t}}
\newcommand{\bfu}{\mbf{u}}

\newcommand{\bfv}{\mbf{v}}

\newcommand{\bfw}{\mbf{w}}
\newcommand{\bfx}{\mbf{x}}
\newcommand{\bfxTil}{\tilde{\mbf{x}}}
\newcommand{\bfX}{\mbf{X}}
\newcommand{\bfXTil}{\tilde{\mbf{X}}}
\newcommand{\bfy}{\mbf{y}}
\newcommand{\bfz}{\mbf{z}}
\newcommand{\bfzTil}{\tilde{\mbf{z}}}
\newcommand{\bfzHat}{\hat{\mbf{z}}}

\newcommand{\vct}[1]{\mathbf{#1}}

\newcommand{\laggamma}{\bfeta}
\newcommand{\lagmu}{\bfnu}
\newcommand{\lagxi}{\xi}
\newcommand{\lagtheta}{\bftheta}
\newcommand{\lagzeta}{\zeta}
\newcommand{\overrel}{\rho}
\newcommand{\mcB}{\mathcal{B}}
\newcommand{\mcL}{\mathcal{L}}
\newcommand{\mcE}{\mathcal{E}}
\begin{document}

\title{Decomposition Methods for Large Scale LP Decoding} 

\author{Siddharth Barman
\thanks{S.~Barman is with the Center for the Mathematics of Information, California Inst. of Tech., CA 91125 (barman@caltech.edu).}
, Xishuo Liu
\thanks{X.~Liu is with the Dept.~of Electrical and Computer
  Engineering, University of Wisconsin, Madison, WI 53706
  (xliu94@wisc.edu).}  
, Stark C. Draper
\thanks{S.~C.~Draper is with the Dept.~of Electrical and Computer
  Engineering, University of Wisconsin, Madison, WI 53706
  (sdraper@ece.wisc.edu).}
, Benjamin Recht
\thanks{B.~Recht is with the Dept.~of Computer Sciences, University of
  Wisconsin, Madison, WI 53706 (brecht@cs.wisc.edu).}  
}

\maketitle

\begin{abstract}

When binary linear error-correcting codes are used over symmetric
channels, a relaxed version of the maximum likelihood decoding problem
can be stated as a linear program (LP).  This LP decoder can be used
to decode error-correcting codes at bit-error-rates comparable to
state-of-the-art belief propagation (BP) decoders, but with
significantly stronger theoretical guarantees.  However, LP decoding
when implemented with standard LP solvers does not easily scale to the
block lengths of modern error correcting codes.  In this paper we draw
on decomposition methods from optimization theory, specifically the
Alternating Directions Method of Multipliers (ADMM), to develop
efficient distributed algorithms for LP decoding.  

The key enabling technical result is a ``two-slice'' characterization
of the geometry of the parity polytope, which is the convex hull of
all codewords of a single parity check code.  This new
characterization simplifies the representation of points in the
polytope.  Using this simplification, we develop an efficient
algorithm for Euclidean norm projection onto the parity polytope.  This
projection is required by ADMM and allows us to use LP decoding, with
all its theoretical guarantees, to decode large-scale error correcting
codes efficiently.

We present numerical results for LDPC codes of lengths more than
$1000$.  The waterfall region of LP decoding is seen to initiate at a
slightly higher signal-to-noise ratio than for sum-product BP, however
an error floor is not observed for LP decoding, which is not the case
for BP.  Our implementation of LP decoding using ADMM executes as fast
as our baseline sum-product BP decoder, is fully parallelizable, and
can be seen to implement a type of message-passing with a particularly
simple schedule.
\end{abstract}

\section{Introduction}

While the problem of error correction decoding dates back at least to
Richard Hamming's seminal work in the 1940s~\cite{hamming:bstj50}, the
idea of drawing upon techniques of convex optimization to solve such
problems apparently dates only to Jon Feldman's 2003
Ph.D. thesis~\cite{feldmanThesis, feldmanEtAl:IT05}.  Feldman and his
collaborators showed that, for binary codes used over symmetric
channels, a relaxed version of the maximum likelihood (ML) decoding
problem can be stated as a linear program (LP).  Considering
graph-based low-density parity-check (LDPC) codes, work by Feldman
\emph{et al.} and later authors~\cite{vontobelKoetter:ISITA04}
\cite{vontobelKoetter:ETT06} \cite{taghaviSiegel:IT08}
\cite{wangYedidiaDraper:ISIT09} demonstrates that the bit-error-rate
performance of LP decoding is competitive with that of standard
sum-product (and min-sum) belief propagation (BP) decoding.
Furthermore, LP decoding comes with a certificate of correctness (ML
certificate)\cite{feldmanEtAl:IT05} -- verifying with probability one
when the decoder has found the ML codeword.  And, if a high-quality
expander~\cite{feldmanEtAL:ISIT04, daskalakisEtAl:IT08} or
high-girth~\cite{aroraEtAl:STOC09} code is used, LP decoding is
guaranteed to correct a constant number of bit flips.

A barrier to the adoption of LP decoding is that solving Feldman's
relaxation using generic LP algorithms is not computationally
competitive with BP.  This is because standard LP solvers do not
automatically exploit the rich structure inherent to the linear
program.  Furthermore, unlike BP, standard solvers do not have a distributed
nature, limiting their scalability via parallelized
(and hardware-compatible) implementation.  In this paper we draw upon
large-scale decomposition methods from convex optimization to develop
an efficient, scalable algorithm for LP decoding.  The result is a
suite of new techniques for efficient error correction of modern
graph-based codes, and insight into the elegant geometry of a
fundamental convex object of error-correction, the \emph{parity
polytope}.

A real-world motivation for developing efficient LP decoding
algorithms comes from applications that have extreme reliability
requirements.  While suitably designed LDPC codes decoded using BP can
achieve near-Shannon performance in the ``waterfall'' regime where the
signal-to-noise ratio (SNR) is close to the code's threshold, they
often suffer from an ``error floor'' in the high SNR regime.  This
limits the use of LDPCs in applications such as magnetic recording and
fiber-optic transport networks.  Error floors result from
weaknesses of arrangements in the graphical structure of the code (variously termed
``pseudocodewords,'' ``near-codewords,'' ``trapping sets,''
``instantons,'' ``absorbing sets'' \cite{FreyKoetterVardy:IT01}
\cite{koetterVontobel:turbo03} \cite{mackayPostol:03}
\cite{richardson:allerton03} \cite{dolecekEtAl:jsac09}), from the
sub-optimal BP decoding algorithm, and from the particulars of the
implementation of BP.  Two natural approaches to improving error floor
performance are to design codes with fewer problematic arrangements
\cite{huEtAl:IT05} \cite{tianEtAl:icc03}
\cite{wangDraperYedidia:preprint11} \cite{zhangEtAl:07}
\cite{tanner:IT81} \cite{wangFossorier:ISIT06}.  or to develop
improved decoding algorithms. 
As LP decoders have been observed to have lower error rate at high SNRs compared to BP decoding~\cite{wangYedidiaDraper:ISIT09,yedidiaWangDraper:IT11,burshtein:IT09,burshtein:IT11}, the
approach taken herein is the latter.

A second motivation is that an efficient LP decoder can help to
develop closer and closer approximations of ML decoders.  This is due
to the strong theoretical guarantees associated with LP solvers.  When
the optimum vertex identified by an LP decoder is integer, the ML
certificate property ensures that that vertex corresponds to the ML
codeword.  When the optimum vertex is non-integer (a
``pseudocodeword''), one is motivated to tighten the relaxation to
eliminate the problematic pseudocodeword, and try again.  Various
methods for tightening LP relaxations have been
proposed~\cite{yangFeldman:JSAC06} \cite{draperYedidiaWang:ISIT07}
\cite{taghaviSiegel:IT08}.  In some settings one can regularly attain
ML performance with few additional
constraints~\cite{yedidiaWangDraper:IT11}.

In this paper, we produce a fast decomposition algorithm based on the
{\em Alternating Direction Method of
  Multipliers}~\cite{boydEtAl:FnT10} (ADMM).  This is a classic
technique in convex optimization and has gained a good deal of
popularity lately for solving problems in compressed
sensing~\cite{Afonso11} and MAP inference in graphical
models~\cite{Martins11}.  As we describe below, when we apply the ADMM
algorithm to LP decoding, the algorithm is a message passing algorithm
that bears a striking resemblance to belief propagation.  Variable
nodes update their estimates of their true values based on information
(messages) from parity check and measurement nodes.  The parity check
nodes produce estimated assignments of local variables based on
information from the variable nodes.

To an optimization researcher, our application of ADMM would appear
quite straight forward.  However, our second contribution, beyond a
naive implementation of ADMM, is a very efficient computation of the
estimates at the parity checks.  Each check update requires the
computation of a Euclidean projection onto the aforementioned parity
polytope.  In Section~\ref{Section:parity_poly}, we demonstrate that
this projection can be computed in linearithmic time in the degree of
the check.  This in turn enables us to develop LP decoders with
computational complexity comparable to (and sometimes much faster
than) BP decoders.

The structure of the decoding LP has been examined before in pursuit
of efficient implementation.  The first attempt was by Vontobel and
Koetter~\cite{vontobelKoetter:turbo06,vontobelKoetter:ETT06} where the
authors used a coordinate-ascent method to develop distributed
message-passing type algorithms to solve the LP.  Their method
requires scheduling updates cyclically on all edges in order to
guarantee convergence.  However, when their approach is matched with
an appropriate message-passing schedule, as determined by Burshtein
in~\cite{burshtein:Turbo08,burshtein:IT09}, converge to the optimal
solution can be attained with a computational complexity which scales
linearly with the block length.  Further,
interior-point~\cite{vontobel:ITA08} \cite{wadayama:ISIT08}
\cite{wadayama:ISIT09} \cite{taghaviEtAl:IT11} and
revised-simplex~\cite{liuQuLiuChen:TComm10} approaches have also been
applied.  In a separate approach Yedida \emph{et al.}
in~\cite{yedidiaWangDraper:IT11} introduced ``Difference-Map BP''
decoding which is a simple distributed algorithm that seems to recover
the performance of LP decoding, but does not have convergence
guarantees.

In this paper we frame the LP decoding problem in the template of
ADMM.  ADMM is distributed, has strong convergence guarantees, simple
scheduling, and, in general, has been observed to be more robust than
coordinate ascent. In addition, we do not have to update parameters between iterations in ADMM.  In
Section~\ref{sec.background} we introduce the LP decoding problem and
introduce notation.  We set up the general formulation of ADMM
problems in Section~\ref{sec.ADMMform} and specialize the formulation
to the LP decoding problem. In Section~\ref{Section:parity_poly} we
present our main technical contributions wherein we develop the
efficient projection algorithm.  We present numerical results
in Section~\ref{sec.simulations} and make some final remarks in
Section~\ref{sec.conclusion}.

\section{Background}
\label{sec.background}

In this paper we consider a binary linear LDPC code $\code$ of length
$N$ defined by a $M \times N$ parity-check matrix $\bfH$. Each of the
$M$ parity checks, indexed by $\mathcal{J} = \{1, 2,\ldots, M \}$,
corresponds to a row in the parity check matrix $\bfH$. Codeword
symbols are indexed by the set $\mathcal{I}=\{1,2,\ldots, N \}$. The
neighborhood of a check $j$, denoted by $\Ne_c(j)$, is the set of
indices $i \in \mathcal{I}$ that participate in the $j$th parity
check, i.e., $\Ne_c(j) = \{ i \mid \bfH_{j,i} = 1 \} $. Similarly for
a component $ i \in \mathcal{I}$, $\Ne_v(i) = \{ j \mid \bfH_{j,i}=1
\}$. Given a vector $\bfx \in \{0,1\}^N$, the $j$th parity-check is
said to be satisfied if $\sum_{i \in \Ne_c(j)} x_i $ is even.  In
other words, the set of values assigned to the $x_i$ for $i \in
\Ne_c(j)$ have even parity. We say that a length-$N$ binary vector
$\bfx$ is a codeword, $\bfx \in \code$, if and only if (iff) all
parity checks are satisfied. In a regular LDPC code there is a fixed
constant $d$, such that for all checks $j \in \mathcal{J}$,
$|\Ne_c(j)| = d$. Also for all components $i \in \mathcal{I}$,
$|\Ne_v(i)| $ is a fixed constant. For simplicity of exposition we
focus our discussion on regular LDPC codes.  Our techniques and
results extend immediately to general LDPC codes and to high density parity check codes as well.

To denote compactly the subset of coordinates of $\bfx$ that
participate in the $j$th check we introduce the matrix $\bfP_j$.  The
matrix $\bfP_j$ is the binary $d \times N$ matrix that selects out the
$d$ components of $\bfx$ that participate in the $j$th check. For
example, say the neighborhood of the $j$th check, $\Ne_c(j) = \{i_1,
i_2, \ldots i_d \}$, where $i_1 < i_2 < \ldots < i_d$.  Then, for all
$k \in [d]$ the $(k, i_k)$th entry of $\bfP_j$ is one and the
remaining entries are zero. For any codeword $\bfx \in \code$ and for
any $j$, $\bfP_j \bfx $ is an even parity vector of dimension $d$.  In
other words we say that $\bfP_j \bfx \in \Pa_d$ for all $j \in
\mathcal{J}$ (a ``local codeword'' constraint) where $\Pa_d$ is
defined as
\begin{equation}
\Pa_d = \{ \bfe \in \{0,1\}^d \mid \| \bfe \|_1 \mbox{ is even}\}.
\label{eq.spcConst}
\end{equation}
Thus, $\Pa_d$ is the set of codewords (the codebook) of the length-$d$
single parity-check code.

We begin by describing maximum likelihood (ML) decoding and the LP
relaxation proposed by Feldman \emph{et al}.  Say vector $\bfxTil$ is
received over a discrete memoryless channel described by channel law
(conditional probability) $W : \mathcal{X} \times \tilde{\mathcal{X}}
\rightarrow \mathbb{R}_{\geq 0}$, $\sum_{\tilde{x} \in
  \tilde{\mathcal{X}}} W(\tilde{x} | x) = 1$ for all $x \in
\mathcal{X}$.  Since the development is for binary codes
$|\mathcal{X}| = 2$.  There is no restriction on
$\tilde{\mathcal{X}}$.  Maximum likelihood decoding selects a codeword
$\bfx \in \code$ that maximizes $p_{\bfXTil|\bfX}(\bfxTil | \bfx )$,
the probability that $\bfxTil$ was received given that $\bfx$ was
sent.  For discrete memoryless channel $W$, $p_{\bfXTil|\bfX}(\bfxTil
| \bfx) = \prod_{i \in \mathcal{I}} W(\tilde{x}_i | {x}_i) $.
Equivalently, we select a codeword that maximizes $ \sum_{i \in
  \mathcal{I}} \log W( {\tilde{x}}_i | {x}_i)$. Let $\gamma_i$ be the
negative log-likelihood ratio, $\gamma_i := \log \left[ W(
  {\tilde{x}}_i | 0 ) / W( {\tilde{x}}_i | 1 ) \right]$.  Since $ \log
W( {\tilde{x}}_i | {x}_i) = - \gamma_i {x}_i + \log W( {\tilde{x}}_i |
0) $, ML decoding reduces to determining an $\bfx \in \code $ that
minimizes $\bfgamma^T \bfx = \sum_{i \in \mathcal{I}} \gamma_i
{x}_i$. Thus, ML decoding requires minimizing a linear function over
the set of codewords.\footnote{This derivation applies to all
  binary-input DMCs.  In the simulations of
  Section~\ref{sec.simulations} we focus on the binary-input additive
  white Gaussian noise (AWGN) channel.  To help make the definitions
  more tangible we now summarize how they specialize for the binary
  symmetric channel (BSC) with crossover probability $p$. For the BSC
  $\tilde{x}_i \in \{0,1\}$.  If $\tilde{x}_i = 1$ then $\gamma_i =
  \log [ W( 1 | 0 ) / W( 1 | 1 )] = \log [p/(1-p)] $ and if
  $\tilde{x}_i = 0$ then $\gamma_i = \log [ W( 0| 0 ) / W( 0 | 1 ) ] =
  \log [(1-p)/p]$.}

Feldman \emph{et al.}~\cite{feldmanEtAl:IT05} show that ML decoding is
  equivalent to minimizing a linear cost over the convex hull of all
  codewords.  In other words, minimize $\bfgamma^T \bfx$ subject to
  $\bfx \in {\rm conv}(\code)$.  The feasible region of this program
  is termed the ``codeword'' polytope.  However, this polytope cannot
  be described tractably.  Feldman's approach is first to relax each
  local codeword constraint $\bfP_j \bfx \in \Pa_d$ to $\bfP_j \bfx
  \in \PP_d$ where
\begin{equation}
\PP_d = {\rm conv} ( \Pa_d) = {\rm conv} ( \{ \bfe \in \{0,1\}^d \mid
\| \bfe \|_1 \mbox{ is even}\}). \label{def.parPoly}
\end{equation}
The object $\PP_d$ is called the ``parity polytope''.  It is the
codeword polytope of the single parity-check code (of dimension $d$).
Thus, for any codeword $\bfx \in \code$, $\bfP_j \bfx $ is a vertex of
$\PP_d$ for all $j$.  When the constraints $\bfP_j \bfx \in \PP_d$ are
intersected for all $j \in \mathcal{J}$ the resulting feasible space
is termed the ``fundamental'' polytope.  Putting these ingredients
together yields the LP relaxation that we study:
\begin{align}
\label{eq.feldmanLP}
 \mbox{minimize } & \bfgamma^T \bfx \ \  \mbox{ s.t. }  \ \ \bfP_j \bfx
 \in \PP_d \ \ \forall \ j \in \mathcal{J}.
\end{align}

The statement of the optimization problem in~(\ref{eq.feldmanLP})
makes it apparent that compact representation of the parity polytope
$\PP_d$ is crucial for efficient solution of the LP.  Study of this polytope
 dates back some decades.  In~\cite{jeroslow:DM75}
Jeroslow gives an explicit representation of the parity polytope and
shows that it has an exponential number of vertices and facets in $d$.
Later, in~\cite{yannakakis:JCSS91}, Yannakakis shows that the parity
polytope has \emph{small lift}, meaning that it is the projection of a
polynomially faceted polytope in a dimension polynomial in $d$.
Indeed, Yannakakis' representation requires a quadratic number of
variables and inequalities.  This is one of the descriptions discussed
in~\cite{feldmanEtAl:IT05} to state the LP decoding problem.

Yannakakis' representation of a vector $\bfu \in \PP_d$ consists of
variables $ \mu_s \in [0,1]$ for all even $s \leq d$. Variable $\mu_s$
indicates the contribution of binary (zero/one) vectors of Hamming
weight $s$ to $\bfu$. Since $\bfu$ is a convex combination of
even-weight binary vectors, $\sum_{\textrm{even } s }^d \mu_s = 1$. In
addition, variables $z_{i,s}$ are used to indicate the contribution to
$u_i$, the $i$th coordinate of $\bfu$ made by binary vectors of
Hamming weight $s$.  Overall, the following set of inequalities over
$O(d^2)$ variables characterize the parity polytope
(see~\cite{yannakakis:JCSS91} and~\cite{feldmanEtAl:IT05} for a
proof).
\begin{align*}
& \ 0 \leq u_i  \leq 1 \ \ \ \ \ \ \forall  \ \ \ i \in [d] \\ 
& \ 0  \leq z_{i,s}  \leq \mu_s \ \ \ \forall \ \ \ i \in [d] \\  
& \sum_{\textrm{even } s}^d \mu_s = 1 \\
& \ u_i   = \sum_{\textrm{even } s}^d z_{i,s}  \ \ \ \forall \ \ \ i \in [d] \\
& \sum_{i=1}^d z_{i,s}  = s \mu_s  \ \ \ \forall \ \  s \textrm{  even},  s \leq d. 
\end{align*}

This LP can be solved with standard solvers in polynomial
time. However, the quadratic size of the LP prohibits its solution
with standard solvers in real-time or embedded decoding
applications. In Section~\ref{Section:characterization} we show that
any vector $\bfu \in \PP_d$ can always be expressed as a convex
combination of binary vectors of Hamming weight $r$ and $r+2$ for some
even integer $r$. Based on this observation we develop a new
formulation for the parity polytope that consists of $O(d)$ variables
and constraints. This is a key step towards the development of an
efficient decoding algorithm. Its smaller description complexity also
makes our formulation particularly well suited for high-density codes whose study we leave for future work.

\section{Decoupled relaxation and optimization algorithms}
\label{sec.ADMMform}

In this section we present the ADMM formulation of the LP decoding
problem and summarize our contributions. In Section~\ref{Section:formulation} we introduce the general ADMM
template. We specialize the template to our problem in
Section~\ref{Section:update_steps}.  We state the algorithm in
Section~\ref{Section:admm_decoding} and frame it in the language of
message-passing in Section~\ref{sec.admm_msg_passing}.

\subsection{ADMM formulation}
\label{Section:formulation}

To make the LP~(\ref{eq.feldmanLP}) fit into the ADMM template we
relax $\bfx$ to lie in the hypercube, $\bfx \in [0,1]^N$, and add the
auxiliary ``replica'' variables $\bfz_j \in \mathbb{R}^d$ for all $j
\in \mathcal{J}$.  We work with a decoupled parameterization of the
decoding LP.
\begin{align}
\label{LP:DecodingLP}
\textrm{minimize } &  \ \ \bfgamma^T \bfx \nonumber \\
\textrm{subject to } & \bfP_j \bfx = \bfz_j  \ \ \ \forall \ j \in \mathcal{J}\nonumber\\
 & \bfz_j \in \PP_d \ \ \  \ \forall \ j \in \mathcal{J}  \nonumber \\
 & \bfx \in [0,1]^N.
\end{align}

The alternating direction method of multiplies works with an augmented
Lagrangian which, for this problem, is
\begin{align}
L_\mu(\bfx, \bfz ,\bflambda) := \ & \bfgamma^T \bfx + \sum_{j \in
  \mathcal{J}} \bflambda_j^T ( \bfP_j \bfx - \bfz_j )   +
\frac{\mu}{2} \sum_{j \in \mathcal{J}}\| \bfP_j \bfx - \bfz_j \|_2^2. \label{eq.augLag}
\end{align}
Here $\bflambda_j \in \mathbb{R}^d$ for $j \in \mathcal{J}$ are the
Lagrange multipliers and $\mu >0$ is a fixed penalty parameter. We use
$\bflambda$ and $\bfz$ to succinctly represent the collection of
$\bflambda_j$s and $\bfz_j$s respectively. Note that the augmented
Lagrangian is obtained by adding the two-norm term of the residual to
the ordinary Lagrangian.  The Lagrangian without the augmentation can
be optimized via a dual subgradient ascent
method~\cite{BertsekasConvexBook}, but our experiments with this
approach required far too many message passing iterations for
practical implementation.  The augmented Lagrangian smoothes the dual
problem leading to much faster convergence rates in
practice~\cite{NocedalWrightBook}. For the interested reader, we
provide a discussion of the standard dual ascent method in the
appendix.

Let $\mathcal{X}$ and $\mathcal{Z}$ denote the feasible
regions for variables $\bfx$ and $\bfz$ respectively: $\mathcal{X} =
[0,1]^N$ and we use $\bfz \in \mathcal{Z}$ to mean that $\bfz_1 \times
\bfz_2 \times \ldots \times \bfz_{|\mathcal{J}|} \in \PP_d \times
\PP_d \times \ldots \times \PP_d$, the $|\mathcal{J}|$-fold product of
$\PP_d$.  Then we can succinctly write the iterations of ADMM as
\begin{align*}
\bfx^{k+1} & := \argmin_{\bfx \in \mathcal{X}} L_\mu(\bfx,\bfz^k, \bflambda^k) \\
\bfz^{k+1} & := \argmin_{\bfz \in \mathcal{Z}}  L_\mu(\bfx^{k+1}, \bfz, \bflambda^k ) \\
\bflambda_j^{k+1} & := \bflambda_j^k + \mu  \left( \bfP_j \bfx^{k+1} -\bfz_j^{k+1}  \right).
\end{align*}
The ADMM update steps involve fixing one variable and minimizing the
other. In particular, $\bfx^k$ and $\bfz^k$ are the $k$th iterate and the
updates to the $\bfx$ and $\bfz$ variable are performed in an alternating
fashion.  We use this framework to solve the LP relaxation proposed by
Feldman \emph{et al.} and hence develop a distributed decoding
algorithm.

\subsection{ADMM Update Steps}
\label{Section:update_steps}

The $\bfx$-update corresponds to fixing $\bfz$ and $\bflambda$
(obtained from the previous iteration or initialization) and
minimizing $L_\mu(\bfx,\bfz,\bflambda)$ subject to $ \bfx \in
[0,1]^N$. Taking the gradient of~(\ref{eq.augLag}), setting the result
to zero, and limiting the result to the hypercube $\mathcal{X} =
[0,1]^N$, the $\bfx$-update simplifies to
\begin{align*}
\bfx & = \Proj_{ [0,1]^N } \left( \bfP^{-1} \times \left( \sum_j \bfP_j^T
\left( \bfz_j - \frac{1}{\mu} \bflambda_j \right) - \frac{1}{\mu} \bfgamma
\right) \right),
\end{align*}
where $\bfP = \sum_j \bfP_j^T \bfP_j$ and $\Proj_{ [0,1]^N }(\cdot)$
corresponds to projecting onto the hypercube $[0,1]^N$.  The latter
can easily be accomplished by independently projecting the components
onto $[0,1]$: setting the components that are greater than $1$ equal
to $1$, the components less than $0$ equal to $0$, and leaving the
remaining coordinates unchanged.  Note that for any $j$, $\bfP_j^T
\bfP_j$ is a $N \times N$ diagonal binary matrix with non-zero entries
at $(i,i)$ if and only if $i$ participates in the $j$th parity check
($i \in \Ne_c(j)$). This implies that $\sum_j \bfP_j^T \bfP_j$ is a
diagonal matrix with the $(i,i)$th entry equal to $|\Ne_v(i)|$.  Hence
$ \bfP^{-1} = ( \sum_j \bfP_j^T \bfP_j )^{-1} $ is a diagonal matrix
with $1/| \Ne_v(i) |$ as the $i$th diagonal entry.

Component-wise, the update rule corresponds to taking the average of
the corresponding replica values, $\bfz_j$, adjusted by the the scaled
dual variable, $\bflambda_j/\mu$, and taking a step in the negative
log-likelihood direction. For any $j \in \Ne_v(i)$ let $z_j^{(i)}$
denote the component of $\bfz_j$ that corresponds to the $i$th
component of $\bfx$, in other words the $i$th component of $\bfP_j^T
\bfz_j$. Similarly let $\lambda_j^{(i)}$ be the $i$th component of
$\bfP_j^T \bflambda_j$. With this notation the update rule for the
$i$th component of $\bfx$ is
\begin{align*}
x_i & = \Proj_{[0,1]} \left( \frac{1}{|\Ne_v(i)|} \left( \sum_{j \in
  N_v(i)} \left( z^{(i)}_j - \frac{1}{\mu} \lambda^{(i)}_j \right) -
\frac{1}{\mu} \gamma_i \right) \right).
\end{align*}
Each variable update can be done in parallel.

The $\bfz$-update corresponds to fixing $\bfx$ and $\bflambda$ and
minimizing $L_\mu(\bfx,\bflambda,\bfz)$ subject to $\bfz_j \in \PP_d$
for all $j \in \mathcal{J}$. The relevant observation here is that the
augmented Lagrangian is separable with respect to the $\bfz_j$s and
hence the minimization step can be decomposed (or ``factored'') into
$| \mathcal{J}|$ separate problems, each of which be solved
independently. This decouples the overall problem, making the approach
scalable.

We start from~(\ref{eq.augLag}) and concentrate on the terms that
involve $\bfz_j$.  For each $j \in \mathcal{J}$ the update is to find
the $\bfz_j$ that minimizes
\begin{align*}
 & \frac{\mu}{2} \| \bfP_j \bfx - \bfz_j \|_2^2 - \bflambda_j^T \bfz_j
  \ \ \ \textrm{ s.t.  } \ \ \bfz_j \in \PP_d.
\end{align*}
Since the values of $\bfx$ and $\bflambda$ are fixed, so are $\bfP_j
\bfx$ and $\bflambda_j /\mu$. Setting $\bfv = \bfP_j \bfx +
\bflambda_j /\mu$ and completing the square we get that the desired
update $\bfz_j^{\ast}$ is
\begin{align*}
\bfz_j^{\ast} & = \argmin_{\bfzTil \in \PP_d} \| \bfv -
  \bfzTil \|_2^2.
\end{align*}
The $\bfz$-update thus corresponds to $|\mathcal{J}|$ projections onto
the parity polytope.

\subsection{ADMM Decoding Algorithm}
\label{Section:admm_decoding}

The complete ADMM-based LP decoding algorithm is specified in the
Algorithm~\ref{Algorithm:ADMM} box. We declare convergence when the
following two conditions are satisfied: (i) replicas differ from the
$\bfx$ variables by less than some tolerance $\epsilon > 0$, and (ii)
the value of each replica differs from its value in the previous
iteration by less than $\epsilon$.

A convergence analysis for ADMM is provided
in~\cite{boydEtAl:FnT10}. We base the following analysis
on~\cite[Thm.~1]{wang:12online}.  The output of ADMM decoding,
$\hat{\bfx}$, satisfies the order relation
\begin{equation*}
\bfgamma^T \hat{\bfx} - \bfgamma^T \bfx^* = O\left(\frac{Md \, \mu}{T}\right),
\end{equation*}
where we recall that $\bfgamma$ is the cost vector, $M$ is the number
of checks, $d$ is the check dimension, $\mu$ is the ADMM penalty
parameter; and where $T$ and $\bfx^*$ are, respectively, the number of
iterations and the solution to the LP decoding problem.  Since for
LDPC codes $M = O(N)$,
\begin{equation*}
\bfgamma^T \hat{\bfx} - \bfgamma^T \bfx^* = O\left(\frac{N}{T}\right).
\end{equation*}
This means that for a given $\delta > 0$, ADMM outputs vector
$\hat{\bfx}$ with $O(1)$ iterations, such that $ \bfgamma^T \hat{\bfx}
- \bfgamma^T \bfx^* < N\delta$.  For each iteration, the $\bfx$-update
has $O(N)$ complexity, the $\bfz$-update has $O(M)$ complexity, and
the $\bflambda$-update has $O(M)$ complexity. Combining the above
results, the complexity of ADMM decoding is given by the following
proposition:
\begin{proposition}
Let $\bfx^*$ be a solution of the LP decoding problem. For any $\delta
> 0$, Algorithm~\ref{Algorithm:ADMM} will, in $O(N)$ time, determine a
vector $\hat{\bfx}$ that satisfies the following bound:
\begin{equation*}
 \bfgamma^T \hat{\bfx} - \bfgamma^T \bfx^* < N\delta.
 \end{equation*} 
\end{proposition}
We note that the experiments we present in Sec.~\ref{sec.simulations}
demonstrate that this convergence estimate is frequently conservative;
we often see convergence in a dozen iterations or fewer.  (See
Fig.~\ref{fig:RandomCode_BSC_iter} for experiments regarding iteration
requirements.)

\begin{algorithm}
\caption{Given an $N$-dimensional vector $\bfxTil\in\tilde{\mathcal{X}}^N$, $M\times N$ parity check matrix $\bfH$, and parameters $\mu$ and
  $\epsilon$, solve the decoding LP specified
  in~(\ref{LP:DecodingLP})}
\label{Algorithm:ADMM}
\begin{algorithmic}[1]

\STATE Construct the negative log-likelihood vector $\bfgamma$ based on
received word $\bfxTil$. Construct the $d \times N $ matrix $\bfP_j$ for all $j \in
\mathcal{J}$.

\STATE Initialize $\bfz_j$ and $\bflambda_j$ as the all zeros vector for
all $j \in \mathcal{J}$. Initialize iterate $k = 0$. For simplicity, we only specify iterate $k$ when determining stopping criteria.
\REPEAT
\FORALL{ $i \in \mathcal{I}$ }
\STATE Update \\
$ x_i \! \leftarrow \! \prod_{[0,1]}\!\!
\left( \frac{1}{|\Ne_v(i)|} \!  \left( \! \sum_{j \in \Ne_v(i)} \! \left( \!
z^{(i)}_j \!\! - \!\! \frac{1}{\mu} \lambda^{(i)}_j \! \right) \!\! -
\!\! \frac{1}{\mu} \gamma_i \! \right) \right)$.
\ENDFOR
\FORALL{ $ j \in \mathcal{J} $}

\STATE Set $ \bfv_j  = \bfP_j \bfx + \bflambda_j /\mu $.

\STATE Update $ \bfz_j \leftarrow \Proj_{\PP_d} (\bfv_j) $ where
$\Proj_{\PP_d} (\cdot)$ means project onto the parity polytope.

\STATE Update $\bflambda_j \leftarrow \bflambda_j + \mu \left( \bfP_j \bfx - \bfz_j\right) $.
\ENDFOR
\STATE $ k \! \leftarrow \! k+1$.
\UNTIL{ $ \sum_j { \| \bfP_j \bfx^k - \bfz^k_j \|^2_{2} } < \epsilon^2 Md $ \\and  $\sum_j { \| \bfz^{k}_j - \bfz^{k-1}_j \|^2_{2} } < \epsilon^2 Md$}\\
{\bf return} $\bfx$.
\end{algorithmic}
\end{algorithm}

\subsection{ADMM Decoding as Message Passing Algorithm}
\label{sec.admm_msg_passing}

We now present a message-passing interpretation of the ADMM approach
to LP decoding as presented in Algorithm~\ref{Algorithm:ADMM}.  For
simplicity, we establish this interpretation by identifying messages
passed between variable nodes and check nodes on a Tanner
graph.

We denote by $x_{ij} (k)$ the replica associated with the edge joining
variable node $i\in \mathcal{I}$ and check node $j \in \mathcal{J}$, where $k$
indicates the $k$th iteration.  Note that $x_{ij_1}(k)= x_{ij_2}(k) =
x_i^k$ for all $j_1,j_2 \in \mathcal{J}$, where $x_i^k$ is the value
of $x_i$ at $k$th iteration in Algorithm~\ref{Algorithm:ADMM}.  The
``message'' $\mia (k) := x_{ij} (k) $ is passed from variable node $i$ to check node
$j$ at the beginning of the $k$th iteration.  Incoming messages to
check node $j$ are denoted as $\main (k) := \{\mia(k):i\in
\Ne_c(j)\}$.  The $\mbf{z_j}$ can also be interpreted as the messages
passed from check node $j$ to the variable nodes in $\Ne_c(j)$,
denoted as $\maout (k) := \{\mai(k):i\in \Ne_c(j)\}$.  Let
$\mbf{\lambda_j' }:= \mbf{\lambda_j}/\mu$ and $\lai' :=
\lambda^{(i)}_j/\mu$.  Then, for all $j \in \Ne_v(i)$
\begin{equation*}
\mia(k+1) = \Proj_{[0,1]} \left(
\frac{1}{|\Ne_v(i)|}\sum_{j\in\mathcal{N}_c(j)}
\left[\mai(k)-\lai'(k)\right] - \frac{\gamma_i}{\mu} \right).
\end{equation*}
The $\mbf{z}$-update can be rewritten as
\begin{equation*}
\maout(k+1) = \Proj_{\PP_d}\left(\main(k) +
\mbf{\lambda'_j}(k)\right).
\end{equation*}
The $\lambda'_j$ update is
\begin{equation*}
\mbf{\lambda'_j}(k+1) = \mbf{\lambda'_j}(k) + \left(\main(k) -
\maout(k)\right).
\end{equation*} 
With this interpretation, it is clear that the ADMM algorithm
decouples the decoding problem and can be performed in a distributed
manner. 

Another nice way to see the decoupling of the decoding problem, and
the connection to message passing, is to use the normal factor graph
formalism~\cite{forney:IT01}.  Then, in a manner similar to that taken
in~\cite[Sec.~III-A]{yedidiaWangDraper:IT11}, replicas are associated
with edges of the normal graph and one creates a dynamics of replicas,
alternately trying to satisfy equality and parity-check constraints.
The dynamics of ADMM are, in general, distinct from those of the
``Divide-and-Concur'' algorithm studied
in~\cite{yedidiaWangDraper:IT11} but related, as recently shown
in~\cite{yedidiaPrincetonTalk}.

\section{The geometric structure of $\PP_d$, and efficient projection onto $\PP_d$}
\label{Section:parity_poly}

In this section we develop our efficient projection algorithm.  Recall
that $\Pa_d =\left\{ \bfe \in \{0,1\}^d \mid \ \| \bfe \|_1 \textrm{
  is even} \right\}$ and that $\PP_d = \mathrm{conv}(\Pa_d)$.
Generically we say that a point $\bfv \in\PP_d$ if and only if there
exist a set of $\bfe_i \in \Pa_d$ such that $\bfv = \sum_i \alpha_i
\bfe_i$ where $\sum_i \alpha_i = 1$ and $\alpha_i \geq 0$.  In
contrast to this generic representation, the initial objective of this
section is to develop a novel ``two-slice'' representation of any
point $\bfv \in \PP_d$: namely that any such vector can be written as
a convex combination of vectors with Hamming weight $r$ and $r+2$ for
some even integer $r$.  We will then use this representation to
construct an efficient projection.

We open the section in Section~\ref{sec.discussion} by describing the
structured geometry of $\PP_d$ and laying out the
results that will follow in ensuing sections.  In
Section~\ref{Section:characterization}, we prove a few necessary
lemmas illustrating some of the symmetry structure of the parity
polytope.  In Section~\ref{sec.parityOfProj} we develop the two-slice
representation and connect the $\ell_1$-norm of the projection of any
$\bfv \in \mathbb{R}^d$ onto $\PP_d$ to the (easily computed)
``constituent parity'' of the projection of $\bfv$ onto the unit
hypercube.  In Section~\ref{sec.projection} we present the projection
algorithm.

\subsection{Introduction to the geometry of $\PP_d$}
\label{sec.discussion}

In this section we discuss the geometry of $\PP_d$.  We develop
intuition and foreshadow the results to come.  We start by making a
few observations about $\PP_d$.  

\begin{itemize}
\item First, we can classify the vertices of $\PP_d$ by their weight.
  We do this by defining $\Pa_d^r$, the constant-weight analog of
  $\Pa_d$, to be the set of weight-$r$ vertices of $\PP_d$:
  \begin{equation}
    \Pa_d^r = \{ {\bf e} \in \{0,1\}^d \, | \, \|\bfe\|_1 = r\},
  \end{equation}
  i.e., the constant-weight-$r$ subcode of $\Pa_d$.  Since all
  elements of $\Pa_d$ are in some $\Pa_d^r$ for some even $r$, $\Pa_d
  = \cup_{0 \leq r \leq d \,: \, r \, \mathrm{even}} \Pa_d^r$.  This
  gives us a new way to think about characterizing the parity
  polytope,
\begin{equation*}
\PP_d = \textrm{conv}(\cup_{0 \leq r \leq d \, : \, r \,
  \mathrm{even}} \ \Pa_d^r).
\end{equation*}
\item Second, we define $\PP_d^r$ to be the convex hull of $\Pa_d^r$,
  \begin{equation}
\PP_d^r = \textrm{conv}(\Pa_d^r) = \mathrm{conv}(\{ \bfe \! \in \!
\{0,1\}^d \mid \| \bfe \|_1 \!=\!  r \}). \label{def.PPdr}
\end{equation}
This object is a ``permutahedron'', so termed because it is the convex
hull of all permutations of a single vector; in this case a length-$d$
binary vector with $r$ ones.  Of course,
\begin{equation*}
\PP_d = \textrm{conv}(\cup_{0 \leq r \leq d \, : \, r \,
  \mathrm{even}} \PP_d^r).
\end{equation*}
\item Third, define the affine hyper-plane consisting of all vectors
  whose components sum to $r$ as 
\begin{equation*}
\mathcal{H}_d^r = \{{\bf x} \in \mathbb{R}^d | {\bf 1}^T {\bf x} = r\}
\end{equation*} 
where ${\bf 1}$ is the length-$d$ all-ones vector.  We can visualize
$\PP_d^r$ as a ``slice'' through the the parity polytope defined as
the intersection of $\mathcal{H}_d^r$ with $\PP_d$.  In other words, a
definition of $\PP_d^r$ equivalent to~(\ref{def.PPdr}) is
\begin{equation*}
\PP_d^r = \PP_d \cap \mathcal{H}_d^r,
\end{equation*}
for $r$ an even integer.
\item Finally, we note that the $\PP_d^r$ are all parallel.  This
  follows since all vectors lying in any of these permutahedra are
  orthogonal to ${\bf 1}$.  We can think of the line segment that
  connects the origin to ${\bf 1}$ as the major axis of the parity
  polytope with each ``slice'' orthogonal to the axis.
\end{itemize}

The above observations regarding the geometry of $\PP_d$ are
illustrated in Fig.~\ref{fig:football}.  Our development will be as
follows: First, in Sec.~\ref{Section:characterization} we draw on a
theorem from~\cite{marshall:2009inequalities} about the geometry of
permutahedra to assert that a point ${\bf v} \in \mathbb{R}^d$ is in
$\PP_d^r$ if and only if a sorted version of ${\bf v}$ is {\em
  majorized} (see Definition~\ref{def.majorize}) by the length-$d$
vector consisting of $r$ ones followed by $d-r$ zeros (the sorted
version of any vertex of $\PP_d^r$).  This allows us to characterize
the $\PP_d^r$ easily.

\begin{figure}
\psfrag{&A}{\scriptsize{$(11111)$}}
\psfrag{&B}{\scriptsize{$(10111)$}}
\psfrag{&C}{\scriptsize{$(01111)$}}
\psfrag{&D}{ \color{blue}$\PP_5^4$}
\psfrag{&E}{ \color{blue}$\PP_5^2$}
\psfrag{&F}{{ \color{red}$\left(\frac{2}{5}\frac{2}{5}\frac{2}{5}\frac{2}{5}\frac{2}{5}\right)$}}
\psfrag{&G}{\scriptsize{$(11000)$}}
\psfrag{&H}{\scriptsize{$(10100)$}}
\psfrag{&I}{\scriptsize{$(00000)$}}
\psfrag{&J}{{ \color{red}$\left(\frac{4}{5}\frac{4}{5}\frac{4}{5}\frac{4}{5}\frac{4}{5}\right)$}}
\psfrag{&K}{\scriptsize{$(11011)$}}
\psfrag{&L}{\scriptsize{$(11101)$}}
\psfrag{&M}{\scriptsize{$(11110)$}}
\psfrag{&N}{\scriptsize{$(10100)$}}

    \begin{center}
    \includegraphics[width=8cm]{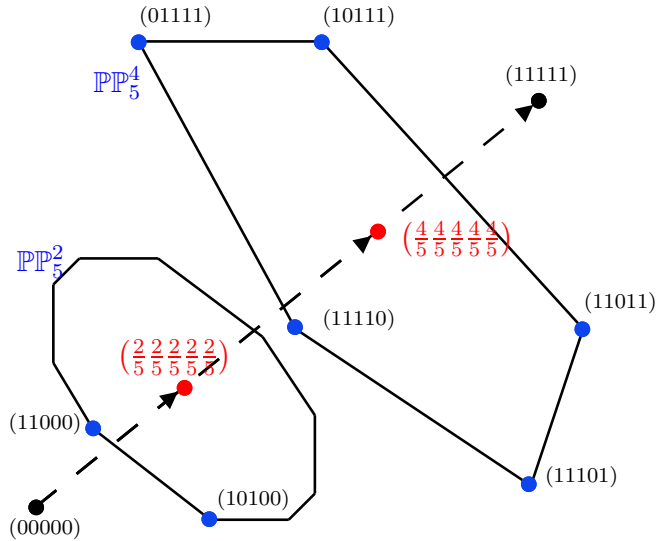}
    \end{center}
    \caption{The parity polytope $\PP_d$ can be expressed as the
      convex hull of ``slices'' through $\PP_d$, each of which
      contains all weight-$r$ vertices.  These sets, $\PP_d^r$ are
      permutahedra.  They are all orthogonal to the line segment
      connecting the origin to the all-ones vector.  The geometry is
      sketched for $d=5$.}
      \label{fig:football}
\end{figure}

Second, we rewrite any point ${\bf u} \in \PP_d$ as, per our second
point above, a convex combination of points in slices of different
weights $r$.  In other words ${\bf u} = \sum_{0 \leq r \leq d \, : \,
  r \, \mathrm{even}} \alpha_r {\bf u}_r$ where ${\bf u}_r \in
\PP_d^r$ and the $\alpha_r$ are the convex weightings.  We develop a
useful characterization of $\PP_d$, the ``two-slice''
Lemma~\ref{lemma:two_slice}, that shows that two slices always
suffices.  In other words we can always write ${\bfu} = \alpha {\bfu}_r + (1-\alpha) {\bfu}_{r+2}$ where ${\bfu}_r \in \PP_d^r$,
${\bfu}_{r+2} \in \PP_d^{r+2}$, $0 \leq \alpha \leq 1$, and $r =
\lfloor \| \bfu \|\rfloor_{\rm even}$, where $\lfloor a
\rfloor_{\textrm{even}}$ is the largest even integer less than or
equal to $a$.  We term the lower weight, $r$, the ``constituent''
parity of the vector.

Third, in Sec.~\ref{sec.parityOfProj} we show that given a point ${\bf
  v} \in \mathbb{R}^d$ that we wish to project onto $\PP_d$, it is
easy to identify the constituent parity of the projection.  To express
this formally, let $\Proj_{\PP_d}(\bfv)$ be the projection of $\bfv$
onto $\PP_d$.  Then, our statement is that we can easily find the even
integer $r$ such that $\Proj_{\PP_d}(\bfv)$ can be expressed as a
convex combination of vectors in $\PP_d^r$ and $\PP_d^{r+2}$.

Finally, in Sec.~\ref{sec.projection} we develop our projection
algorithm.  In short, our approach is as follows: Given a vector $\bfv
\in \mathbb{R}^d$ we first compute $r$, the constituent parity of its
projection.  Given the two-slice representation, projecting onto
$\PP_d$ is equivalent to determining an $\alpha \in [0,1]$, a vector
$\bfa \in \PP_d^r$, and a vector $ \bfb \in \PP_d^{r+2} $ such that
the $\ell_2$ norm of $ \bfv - \alpha \bfa - (1-\alpha) \bfb$ is
minimized.

In~\cite{barmanEtAl:allerton11} we showed that, given $\alpha$, this
projection can be accomplished in two steps.  We first scale $\PP_d^r$
by the convex weighting parameter $\alpha$ to obtain $\alpha \PP_d^r =
\{{\bf x} \in \mathbb{R}^d | 0 \leq x_i \leq \alpha, \sum_{i=1}^d x_i
= \alpha r\}$ and project $\bfv$ onto $\alpha \PP_d^r$. Then we
project the residual onto $(1-\alpha) \PP_d^r$.  The object $\alpha
\PP_d^r$ is an $\ell_1$ ball with box constraints.  Projection onto
$\alpha \PP_d^r$ can be done efficiently using a type of waterfilling.
Since the function $ \min_{\bfa \in \PP_d^r \textrm{, } \bfb \in
  \PP_d^{r+2} } \| \bfv - \alpha \bfa - (1-\alpha) \bfb \|^2_2$ is
convex in $\alpha$ we can perform perform a one-dimensional line
search (using, for example, the secant
method~\cite[p.~188]{AllenNumericalBook}) to determine the optimal value for
$\alpha$ and thence the desired projection.

In contrast to the original approach, in Section~\ref{sec.projection}
we develop a far more efficient algorithm that avoids the pair of
projections and the search for $\alpha$. In particular, taking
advantage of the convexity in $\alpha$ we use majorization to
characterize the convex hull of $\PP_d^r$ and $\PP_d^{r+2}$ in terms
of a few linear constraints (inequalities). As projecting onto the
parity polytope is equivalent to projecting onto the convex hull of
the two slices, we use the characterization to express the projection
problem as a quadratic program, and develop an efficient method that
directly solves the quadratic program.  Avoiding the search over
$\alpha$ yields a considerable speed-up over the original approach
taken in~\cite{barmanEtAl:allerton11}.

\subsection{Permutation Invariance of the Parity Polytope and Its Consequences}
\label{Section:characterization}

Let us first describe some of the essential features of the parity
polytope that are critical to the development of our efficient
projection algorithm.  First, note the following
\begin{proposition}\label{prop:member-permute}
 $\bfu \in \PP_d$ if and only if $\mathbf{\Sigma} \bfu$ is in the
parity polytope for every $d\times d$ permutation matrix $\mathbf{\Sigma}$.  
\end{proposition}

This proposition follows immediately because the vertex set
$\mathbb{P}_d$ is invariant under permutations of the coordinate axes.

Since we will be primarily concerned with projections onto the parity
polytope, let us consider the optimization problem
\begin{equation}
  \textrm{minimize}_{\bfz} \| \bfv - \bfz \|_2~~~\textrm{subject to}~ \bfz \in\PP_d\,.
\end{equation}
The optimal $\bfz^*$ of this problem is the Euclidean projection of
$\vct{v}$ onto $\PP_d$, which we denote by $\bfz^*=\Pi_{\PP_d}(\bfv)$.
Again using the symmetric nature of $\PP_d$, we can show the useful
fact that if $\bfv$ is sorted in descending order, then so is
$\Pi_{\PP_d}(\bfv)$.

\begin{proposition}\label{prop:proj-permute}
Given a vector $\bfv \in \mathbb{R}^d$, the component-wise ordering of
$\Proj_{\PP_d}(\bfv)$ is same as that of $\bfv$.
\end{proposition}
\begin{proof}
We prove the claim by contradiction. Write $\bfz^* =
\Proj_{\PP_d}(\bfv)$ and suppose that for indices $i$ and $j$ we have
$v_i > v_j $ but $z_i^* < z_j^*$.  Since all permutations of $\bfz^*$
are in the parity polytope, we can swap components $i$ and $j$ of
$\bfz^*$ to obtain another vector in $\PP_d$. Under the assumption
$z_j^* > z_i^*$ and $ v_i - v_j > 0$ we have $z_j^*( v_i - v_j) >
z_i^*( v_i - v_j)$. This inequality implies that $(v_i - z^*_i)^2 +
(v_j - z^*_j)^2 > (v_i - z^*_j)^2 + (v_j - z^*_i)^2 $, and hence we
get that the Euclidean distance between $\bfv$ and $\bfz^*$ is greater
than the Euclidean distance between $\bfv$ and the vector obtained by
swapping the components.
\end{proof}

These two propositions allow us assume through the remainder of this
section that our vectors are presented sorted in descending order
unless explicitly stated otherwise.

The permutation invariance of the parity polytope also lets us also
employ powerful tools from the theory of majorization to simplify
membership testing and projection.  The fundamental theorem we exploit
is based on the following definition.

\begin{definition}\label{def.majorize}
Let $\bfu$ and $\bfw$ be $d$-vectors sorted in decreasing order.  The
vector $\bfw$ \emph{majorizes} $\bfu$ if
\begin{align*}
	\sum_{k=1}^q u_k &\leq \sum_{k=1}^q w_k \quad \forall~1\leq
        q<d,\\ \sum_{k=1}^d u_k &= \sum_{k=1}^d w_k\,.
\end{align*}
\end{definition}

Our results rely on the following Theorem, which states that a vector
lies in the convex hull of all permutations of another vector if and
only if the former is majorized by the latter (see~\cite{marshall:2009inequalities} and references therein).
\begin{theorem}
\label{theorem:major_hull}
Suppose $\bfu$ and $\bfw$ are $d$-vectors sorted in decreasing order.  Then
$\bfu$ is in the convex hull of all permutations of $\bfw$ if and only if
$\bfw$ majorizes $\bfu$. 
\end{theorem}

To gain intuition for why this theorem might hold, suppose that $\bfu$
is in the convex hull of all of the permutations of $\bfw$.  Then
$\bfu = \sum_{i=1}^n p_i \mathbf{\Sigma}_i \bfw$ with $\mathbf{\Sigma}_i$ being
permutation matrices, $p_i \geq 0$, and $\vct{1}^T \vct{p}=1$.  The
matrix $\bfQ = \sum_{i=1}^n p_i \mathbf{\Sigma}_i$ is doubly stochastic, and
one can immediately check that if $\bfu = \bfQ \bfw$ and $\bfQ$ is
doubly stochastic, then $\bfw$ majorizes $\bfu$.

To apply majorization theory to the parity polytope, begin with one of
the permutahedra $\PP_d^s$.  We recall that $\PP_d^s$ is equal to the
convex hull of all binary vectors with weight $s$, equivalently the
convex hull of all permutations of the vector consisting of $s$ ones
followed by $d-s$ zeros.  Thus, by Theorem~\ref{theorem:major_hull},
$\bfu \in [0,1]^d$ is in $\PP_d^s$ if and only if
\begin{align}
	\sum_{k=1}^q u_k &\leq \min (q,s) \quad \forall \ 1\leq
        q<d, \label{eq.relOne}\\ \sum_{k=1}^d u_k &=
        s. \label{eq.relTwo}
\end{align}

The parity polytope $\PP_d$ is simply the convex hull of all of the $\PP_d^s$ with $s$
even.  Thus, we can use majorization to provide an alternative
characterization of the parity polytope to that of Yannakakis or Jeroslow.

\begin{lemma}
\label{lemma:majorize}
A sorted vector $\bfu \in \PP_d$ if and only if there exist non-negative
coefficients $\{ \mu_s \}_{\textrm{even } s \leq d}$ such that
\begin{align}
\sum_{s~\mathrm{even}}^d \mu_s & = 1, \quad
\mu_s \geq 0.\label{eq:conv-comb}\\
\sum_{k=1}^q u_k & \leq \sum_{s~\mathrm{even}}^d \mu_s \min(q,s) \quad
\forall~1\leq q<d \label{eq:majorize_bnd} \\ 
\sum_{k=1}^d u_k &= \sum_{s~\mathrm{even}}^d \mu_s s. \label{eq:majorize_eq}\\ \nonumber
\end{align}
\end{lemma}

\begin{proof}
First, note that every vertex of $\PP_d$ of weight $s$ satisfies these inequalities
with $\mu_s=1$ and $\mu_{s'}=0$ for $s'\neq s$.  Thus $\bfu \in \PP_d$
must satisfy \eq{conv-comb}-\eq{majorize_eq}.  Conversely, if $\bfu$
satisfies \eq{conv-comb}-\eq{majorize_eq},  then $\bfu$ is majorized by the vector
\[
\bfw = \sum_{s~\mathrm{even}}^d \mu_s \bfb_s
\]
where $\bfb_s$ is a vector consisting of $s$ ones followed by $d-s$
zeros.  $\bfw$ is contained in $\PP_d$ as are all of its
permutations.  Thus,  we conclude that $\bfu$ is also contained in $\PP_d$.
\end{proof}

While Lemma~\ref{lemma:majorize} characterizes the containment of a
vector in $\PP_d$, the relationship is not one-to-one; for a
particular $\bfu \in \PP_d$ there can be many sets $\{\mu_s\}$ that
satisfy the lemma.  We will next show that there is always one assignment of
$\mu_s$ with only two non-zero $\mu_s$.

\subsection{Constituent Parity of the Projection}
\label{sec.parityOfProj}

For $a\in \mathbb{R}$, let $\lfloor a \rfloor_{\textrm{even}}$ denote the
``even floor'' of $a$, i.e., the largest even integer $r$ such that $r
\leq a$.    Define the ``even-ceiling,'' $\lceil a \rceil_{\textrm{even}}$ similarly.
For a vector $\bfu$ we term $\lfloor \| \bfu \|_1
\rfloor_{\textrm{even}}$ the {\em constituent parity} of vector
$\bfu$.  In this section we will show that if $\bfu \in \PP_d$ has
constituent parity $r$, then it can be written as a convex combination
of binary vectors with weight equal to $r$ and $r+2$. This result is
summarized by the following lemma:

\begin{lemma} (``Two-slice'' lemma)
\label{lemma:two_slice}
A vector $\bfu \in \mathbb{PP}_d$ iff $\bfu$ can be
expressed as a convex combination of vectors in $\PP_d^r$ and
$\PP_d^{r+2}$ where $r= \lfloor \|\bfu\|_1 \rfloor_{\textrm{even}}$.
\end{lemma}

\begin{proof}  Consider any (sorted) $\bfu \in \PP_d$.  Lemma~\ref{lemma:majorize} 
tells us that there is always (at least one) set $\{\mu_s\}$ that
satisfy~(\ref{eq:conv-comb})--(\ref{eq:majorize_eq}).  Letting $r$ be
defined as in the lemma statement, we define $\alpha$ to be the unique
scalar between zero and one that satisfies the relation $\|\bfu \|_1 =
\alpha r + (1-\alpha) (r+2)$:
\begin{equation}
\alpha = \frac{2 + r -\|\bfu\|_1}{2}. \label{eq.choiceOfAlpha}
\end{equation}
Then, we choose the following candidate assignment: $\mu_r = \alpha$,
$\mu_{r+2} = 1-\alpha$, and all other $\mu_s = 0$.  We show that this
choice satisfies~(\ref{eq:conv-comb})--(\ref{eq:majorize_eq}) which
will in turn imply that there is a $\bfu_r \in \PP_d^r$ and a $\bfu_{r+2}
\in \PP_d^{r+2}$ such that $\bfu = \alpha \bfu_r + (1-\alpha)
\bfu_{r+2}$.

First, by the definition of $\alpha$, (\ref{eq:conv-comb})
and~(\ref{eq:majorize_eq}) are both satisfied.  Further, for the
candidate set the relations~(\ref{eq:majorize_bnd})
and~(\ref{eq:majorize_eq}) simplify to
\begin{align}
\sum_{k=1}^q u_k & \leq \alpha \min (q,r) \! + \! (1\!-\!\alpha)\min
(q,r\!+\!2), \ \forall \ 1\!\leq \! q \! < \! d,
\label{eq:simple_majorize_bnd}\\
\label{eq:simple_majorize_eq}	\sum_{k=1}^d u_k & = \alpha r + (1-\alpha) (r+2).
\end{align}
To show that~(\ref{eq:simple_majorize_bnd}) is satisfied is
straightforward for the cases $q \leq r$ and $q \geq r+2$.  First
consider any $q \leq r$. Since $\min(q, r) = \min(q, r+2) = q$,
$u_k\leq 1$ for all $k$, and there are only $q$ terms,
(\ref{eq:simple_majorize_bnd}) must hold.  Second, consider any $q
\geq r+2$.  We use~(\ref{eq:simple_majorize_eq}) to write
$\sum_{k=1}^q u_k = \alpha r + (1-\alpha) (r+2) - \sum_{q+1}^d u_k$.
Since $u_k \geq 0$ this is upper bounded by $\alpha r + (1-\alpha)
(r+2)$ which we recognize as the right-hand side
of~(\ref{eq:simple_majorize_bnd}) since $r = \min(q,r)$ and $r+2 =
\min(q,r+2)$.

It remains to verify only one more inequality
in~\eq{simple_majorize_bnd} namely the case when $q = r+1$, which is
\begin{equation*}
\sum_{k=1}^{r+1} u_k \leq \alpha r + (1-\alpha)(r+1) = r+ 1-\alpha.
\end{equation*}
To show that the above inequality holds, we maximize the
right-hand-side of \eq{majorize_bnd} across {\em all} valid choices of
$\{\mu_s\}$ and show that the resulting maximum is exactly $r+
1-\alpha$.  Since this maximum is attainable by some choice of
$\{\mu_s\}$ and our choice meets that bound, our choice is a valid
choice.

The details are as follows: Since $\bfu \in \PP_d$ any valid choice for
$\{\mu_s\}$ must satisfy \eq{conv-comb} which, for $q = r+1$, is
\begin{equation}\label{eq:majorize-bnd-rp1}
\sum_{k=1}^{r+1} u_k \leq \sum_{s~\mathrm{even}}^d \mu_s   \min (s,r+1).
\end{equation}
To see that across {\em all} valid choice of $\{\mu_s\}$ the largest
value attainable for the right hand side is precisely $r + 1-\alpha$
consider the linear program
\begin{equation*}\label{eq:primal-lp}
	\begin{array}{ll} \mbox{maximize} &  \sum_{s~\mathrm{even}} \mu_s   \min(s,r+1)\\
	\mbox{subject to} & \sum_{s~\mathrm{even}} \mu_s=1\\ 
& \sum_{s~\mathrm{even}} \mu_s s = \alpha r + (1-\alpha)(r+2) \\ 
&\mu_s\geq 0.
	\end{array}
\end{equation*}
The first two constraints are simply~(\ref{eq:conv-comb})
and~(\ref{eq:majorize_eq}).  Recognizing $\alpha r + (1-\alpha)(r+2) =
r+2-2\alpha$, the dual program is
\begin{equation*}
	\begin{array}{ll} \mbox{minimize} &  (r+2-2\alpha) \lambda_1 + \lambda_2\\
	\mbox{subject to} & \lambda_1 s + \lambda_2 \geq \min(s,r+1)~\forall~s~\mbox{even}.
	\end{array}	
\end{equation*}
Setting $\mu_{r}=\alpha$, $\mu_{r+2}=(1-\alpha)$, the other primal
variable to zero, $\lambda_1 = 1/2$, and $\lambda_2=r/2$, satisfies
the Karush-Kuhn-Tucker (KKT) conditions for this primal/dual pair of
LPs.  The associated optimal cost is $r+1-\alpha$.  Thus, the right
hand side of~\eq{majorize-bnd-rp1} is at most $r+1-\alpha$.

We have proved that if $\bfu \in \PP_d$ then the choice of $r =
\lfloor \|\bfu\|_1 \rfloor_{\textrm{even}}$ and $\alpha$ as
in~(\ref{eq.choiceOfAlpha}) satisfies the requirements of
Lemma~\ref{lemma:majorize} and so we can express $\bfu$ as $\bfu =
\alpha \bfu_r + (1-\alpha) \bfu_{r+2}$.  The converse---given a
vector $\bfu$ that is a convex combination of vectors in $\PP_d^r$ and
$\PP_d^{r+2}$ it is in $\PP_d$---holds because $\textrm{conv}(\PP_d^r
\cup \PP_d^{r+2}) \subseteq \PP_d$.
\end{proof}

A useful consequence of Theorem~\ref{theorem:major_hull} is the following
corollary.
\begin{corollary}
\label{corollary:even_containment}
Let $\bfu$ be a vector in $[0,1]^d$.  If $\sum_{k=1}^d u_k$ is an even
integer then $ \bfu \in \PP_d$.
\end{corollary}
\begin{proof}
Let $\sum_{k=1}^d u_k = s$. Since $\bfu$ is majorized by a sorted
binary vector of weight $s$ then, by Theorem~\ref{theorem:major_hull},
$\bfu \in \PP_d^s$ which, in turn, implies $\bfu \in \PP_d$.
\end{proof}

We conclude this section by showing that we can easily compute the
constituent parity of $\Pi_{\PP_d}(\bfv)$ without explicitly computing
the projection of $\bfv$.

\begin{lemma}
\label{lemma:const_parity}
For any vector $\bfv \in \mathbb{R}^d$, let $\bfz=
\Proj_{[0,1]^d}(\bfv)$, the projection of $\bfv$ onto $[0,1]^d$ and
denote by $\Proj_{\PP_d} (\bfv) $ the projection of $\bfv$ onto the
parity polytope.  Then
\begin{equation*}
\lfloor \| \bfz \|_1 \rfloor_{\textrm{even}} \leq \| \Proj_{\PP_d} (\bfv) \|_1 \leq
\lceil \| \bfz \|_1 \rceil_{\textrm{even}}\,.
\end{equation*}
\end{lemma}

That is, we can compute the constituent parity of the projection of
$\bfv$ by projecting $\bfv$ onto $[0,1]^d$ and computing the even
floor.

\begin{proof}
Let $\rho_U= \lceil \| \bfz \|_1 \rceil_{\textrm{even}}$ and $
\rho_L = \lfloor \| \bfz \|_1 \rfloor_{\textrm{even}}$. We prove the
following fact: given any $\bfy' \in \PP_d$ with $\| \bfy' \|_1 > \rho_U$
there exists a vector $\bfy \in [0,1]^d$ such that $\| \bfy \|_1 = \rho_U$,
$\bfy \in \PP_d$, and $ \| \bfv - \bfy \|_2^2 < \| \bfv - \bfy'\|^2_2$.  The
implication of this fact will be that any vector in the parity
polytope with $\ell_1$ norm strictly greater that $\rho_U$ cannot be
the projection of $\bfv$.  Similarly we can also show that any vector
with $\ell_1$ norm strictly less than $\rho_L$ cannot be the
projection on the parity polytope.

First we construct the vector $\bfy$ based on $\bfy'$ and $\bfz$.  Define the
set of ``high'' values to be the coordinates on which $y'_i$ is
greater than $z_i$, i.e., $\mathcal{H} := \{ i \in [d] \mid y'_i >
z_i \}$.  Since by assumption $\|\bfy' \|_1 > \rho_U \geq \| \bfz
\|_1$ we know that $|\mathcal{H}| \geq 1 $. Consider the test vector
$\bft$ defined component-wise as
\begin{align*}
t_i = \begin{cases} z_i \textrm{ if } i \in \mathcal{H}, \\ y'_i
  \textrm{ otherwise. } \end{cases}
\end{align*}
Note that $\| \bft \|_1 \leq \| \bfz \|_1 \leq \rho_U < \| \bfy'
\|_1$. The vector $\bft$ differs from $\bfy'$ only in $\mathcal{H}$.
Thus, by changing (reducing) components of $\bfy'$ in the set
$\mathcal{H}$ we can obtain a vector $\bfy$ such that $\| \bfy \|_1 =
\rho_U$.  In particular there exists a vector $\bfy$ with $\| \bfy
\|_1 = \rho_U$ such that $ y'_i \geq y_i \geq z_i$ for $i \in
\mathcal{H}$ and $y_i = y'_i$ for $i \notin \mathcal{H}$. Since the
$\ell_1$ norm of $\bfy$ is even and it is in $[0,1]^d$ we have by
Corollary~\ref{corollary:even_containment} that $\bfy \in \PP_d$.

We next show that for all $i \in \mathcal{H}$, $ |v_i - y_i| \leq |
v_i - y'_i| $.  The inequality will be strict for at least one $i$
yielding $\| \bfv - \bfy \|_2^2 < \| \bfv - \bfy' \|_2^2 $ and thereby proving the
claim.  

We start by noting that $\bfy' \in \PP_d$ so $y'_i \in [0,1]$ for all
$i$. Hence, if $z_i < y'_i$ for some $i$ we must also have $z_i < 1$,
in which case $v_i \leq z_i$ since $z_i$ is the projection of $v_i$
onto $[0,1]$. In summary, $z_i <1$ iff $v_i <1$ and when $z_i < 1$
then $v_i \leq z_i$.  Therefore, if $y'_i > z_i$ then $z_i \geq
v_i$. Thus for all $i \in \mathcal{H}$ we get $ y'_i \geq y_i \geq z_i
\geq v_i$ where the first inequality is strict for at least one $i$.
Since $y_i = y'_i$ for $ i \notin \mathcal{H}$ this means that $|v_i -
y_i| \leq | v_i - y'_i|$ for all $i$ where the inequality is strict
for at least one value of $i$.  Overall, $\| \bfv - \bfy \|_2^2 < \|
\bfv - \bfy' \|_2^2$ and both $\bfy \in \PP_d$ (by construction) and
$\bfy' \in \PP_d$ (by assumption).  Thus, $\bfy'$ cannot be the
projection of $\bfv$ onto $\PP_d$.  Thus the $\ell_1$ norm of the
projection of $\bfv$, $\| \Proj_{\PP_d} (\bfv) \|_1 \leq \rho_U$.  A
similar argument shows that $\| \Proj_{\PP_d} (\bfv) \|_1 \geq \rho_L$
and so $\| \Proj_{\PP_d} (\bfv) \|_1$ must lie in $[\rho_L, \rho_U ]$
\end{proof}

\subsection{Projection Algorithm}
\label{sec.projection}

In this section we formulate a quadratic program (Problem PQP) for the projection problem
and then develop an algorithm (Algorithm~\ref{Algorithm:PPd}) that
efficiently solves the quadratic program.

Given a vector $\bfv \in \mathbb{R}^d$, set $r = \lfloor \|
\Pi_{[0,1]^d}(\bfv) \|_1 \rfloor_{\textrm{even}}$.  From Lemma
\ref{lemma:const_parity} we know that the constituent parity of
$\bfz^*:=\Proj_{\PP_d}(\bfv)$ is $r$.  We also know that if $\bfv$ is
sorted in descending order then $\bfz^*$ will also be sorted in
descending order.  Let $\bfS$ be a $(d-1) \times d$ matrix with
diagonal entries set to $1$, $\bfS_{i,i+1}=-1$ for $1 \leq i \leq
d-1$, and zero everywhere else:
\begin{equation*}
\bfS = \left[ \begin{array}{ccccccc} 1 & -1 & 0 & 0 & \ldots & 0 & 0\\
0 & 1 & -1 & 0 & \ldots & 0 & 0\\
0 & 0 & 1 & -1 & \ldots & 0 & 0\\
\vdots &  &  &  \ddots & \ddots &  & \vdots\\
0 & 0 & 0 & 0 & \ldots & -1 & 0\\
0 & 0 & 0 & 0 & \ldots & 1 & -1\\
\end{array} \right].
\end{equation*}
The constraint that $\bfz^*$ has to be sorted in
decreasing order can be stated as $\bfS \bfz^* \geq \bfzero$, where
$\bfzero$ is the all-zeros vector.

In addition, Lemma~\ref{lemma:two_slice}
implies that $\bfz^*$ is a convex combination of vectors
of Hamming weight $r$ and $r+2$. Using inequality
(\ref{eq:simple_majorize_bnd}) we get that a $d$-vector $\bfz \in
[0,1]^d$, with
\begin{equation}\label{eq:qp-norm}
\sum_{i=1}^d z_i = \alpha r + ( 1 - \alpha) (r+2),
\end{equation}
is a convex combination of vectors of weight $r$ and $r+2$ iff it
satisfies the following bounds:
\begin{align}
\sum_{k=1}^q z_{(k)} &\leq \alpha \min (q,r) + (1\!-\!\alpha)\min (q,r\!+\!2)
\ \ \forall \ \ 1\leq q<d \label{eq:two-slice},
\end{align}
where $z_{(k)}$ denotes the $k$th largest component of $\bfz$. As we
saw in the proof of Lemma~\ref{lemma:majorize}, the fact that the
components of $\bfz$ are no more than one implies that inequalities
(\ref{eq:two-slice}) are satisfied for all $ q \leq r$. Also,
(\ref{eq:qp-norm}) enforces the inequalities for $q \geq
r+2$. Therefore, inequalities in (\ref{eq:two-slice}) for $q \leq r$
and $q \geq r+2$ are redundant.  Note that in addition we can
eliminate the variable $\alpha$ by solving~\eq{qp-norm} giving $\alpha
= 1 + \frac{r - \sum_{k=1}^d z_k }{2}$ (see
also~(\ref{eq.choiceOfAlpha})). Therefore, for a sorted vector $\bfv$,
we can write the projection onto $\PP_d$ as the optimization problem
\begin{align}
\textrm{minimize} & \ \ \ \frac{1}{2} \| \bfv - \bfz \|_2^2 \nonumber
\\ \textrm{subject to} & \ \ \ 0 \leq z_i \leq 1 \ \ \forall \ i
\nonumber \\ 
& \ \ \ \bfS \bfz \geq 0 \nonumber\\
&
\ \ \ 0 \leq \ 1 + \frac{r - \sum_{k=1}^d z_k }{2} \ \leq
1 \label{eq:to-replace-alpha} \\ & \ \ \ \sum_{k=1}^{r+1} z_{k} \leq r
- \frac{r - \sum_{k=1}^d z_k }{2}.  \label{eq:to-def-lambda}
\end{align}

The last two constraints can be simplified as follows: First,
constraint (\ref{eq:to-replace-alpha}) simplifies to $ r \leq
\sum_{k=1}^d z_k \leq r+2$. Next, define the vector
\begin{align}\label{eq:fr-def}
\bff_r = ( \underbrace{1, 1, \ldots, 1}_{r+1} , \underbrace{-1,-1,\ldots,-1}_{d-r-1} )^T.
\end{align}
we can rewrite inequality (\ref{eq:to-def-lambda}) as $\bff_r^T\bfz
\leq r$.  Using these simplifications yields the final form of our quadratic
program:

{\em Problem PQP:}
\begin{align}
\textrm{minimize} & \ \ \ \frac{1}{2} \|\bfv - \bfz \|_2^2 \nonumber \\
\textrm{subject to} &  \ \ \  0 \leq z_i \leq 1 \ \ \forall \ i \label{eq:boxConstraint}\\  
& \ \ \ \bfS \bfz \geq \bfzero \label{eq:orderConstraint} \\
& \ \ \ r \leq \vct{1}^T\bfz  \leq r+2 \label{eq:constitParityConstraint} \\
& \ \ \    \bff_r^T \bfz \leq r.  \label{eq:Lambda}
\end{align}

The projection algorithm we develop efficiently solves the KKT
conditions of PQP.  The objective function is strongly convex and the
constraints are linear.  Hence, the KKT conditions are not only
necessary but also sufficient for optimality.  To formulate the KKT
conditions, we first construct the Lagrangian with dual variables
$\beta$, $\lagmu$, $\laggamma$, $\lagxi$, $\lagtheta$, and $\lagzeta$:
\begin{align*}
\Lag & = \frac{1}{2} \| \bfv - \bfz \|_2^2 - \beta \left(r -
\bff_r^T\bfz \right) - \lagmu^T ( \mathbf{1} - \bfz ) -
\laggamma^T \bfz \nonumber \\ & \ \ \ \ - \lagxi \left( r+2 -
\mathbf{1}^{T} \bfz \right) - \lagzeta( \mathbf{1}^{T} \bfz - r ) -
\lagtheta^T \bfS \bfz\,.
\end{align*}

The KKT conditions are then given by stationarity of the Lagrangian,
complementary slackness, and feasibility.
\begin{align} 
\bfz = \bfv - \beta \bff_r - \lagmu &+ \laggamma - ( \lagxi - \lagzeta)
\mathbf{1} + \bfS^T \lagtheta. \label{eq:kkt-grad}\\
 0 \leq \beta  \ \ \ & \perp \ \ \  \bff_r^T \bfz  - r \leq 0 \nonumber \\
\bfzero \leq \lagmu  \ \ \ & \perp \ \ \  \bfz \leq \mathbf{1} \nonumber \\
\bfzero \leq \laggamma \ \ \ & \perp \ \ \ \bfz \geq \bfzero \nonumber \\  
\bfzero \leq \lagtheta \ \ \ & \perp \ \ \ \bfS \bfz \geq \bfzero \nonumber \\
0 \leq \lagxi  \ \ \ & \perp \ \ \  \mathbf{1}^{T} \bfz - r - 2 \leq 0 \nonumber \\
0 \leq \lagzeta \ \ \ & \perp \ \ \ \mathbf{1}^{T} \bfz - r \geq 0. \nonumber 
\end{align}
A vector $\bfz$ that satisfies (\ref{eq:kkt-grad}) and the following orthogonality
conditions is equal to the projection of $\bfv$ onto $\PP_d$.

To proceed, set $\beta_{\textrm{max}} = \frac{1}{2}[v_{r+1} -
v_{r+2}]$ and define the parameterized vector
\begin{equation}
\bfz(\beta) := \Proj_{[0,1]^d} ( \bfv - \beta \bff_r) \label{def.zProj}\,.
\end{equation} 
The following lemma implies that the optimizer of PQP, i.e., $\bfz^* =
\Proj_{\PP_d}(\bfv)$, is $\bfz(\beta_{\textrm{opt}})$ for some
$\beta_{\textrm{opt}} \in [0, \beta_{\textrm{max}}]$.

\begin{lemma}
\label{lemma:optimality}
There exists a $\beta_{\textrm{opt}} \in [ 0 , \beta_{\textrm{max}}] $
such that $\bfz( \beta_{ \textrm{opt} } ) $ satisfies
the KKT conditions of the quadratic program PQP.
\end{lemma}

\begin{proof}
Note that when $\beta > \beta_{\textrm{max}}$ we have that
$z_{r+1}(\beta)< z_{r+2}(\beta)$ and $\bfz(\beta)$ is ordered
differently from $\bfv$ and $\bff_r^T \bfz(\beta) < r$.  Consequently
$z(\beta)$ cannot be the projection onto $\PP_d$ for
$\beta>\beta_{\mathrm{max}}$.  At the other boundary of the interval,
when $\beta=0$ we have $\bfz(0) =\Proj_{[0,1]^d} (\bfv)$.  If
$\bff_r^T\bfz(0)=r$, then $\bfz(0)\in \PP_d$ by
Corollary~\ref{corollary:even_containment}.  But since $\bfz(0)$ is
the closest point in $[0,1]^d$ to $\bfv$, it must also be the closest
point in $\PP_d$.

Assume now that $\bff_r ^T\bfz(0) > r$.  Taking the directional
derivative with respect to $\beta$ increasing, we obtain the following:
\begin{align}
\frac{\partial \bff_r ^T \bfz (\beta) } {\partial \beta} & =
\bff_r^T \frac{\partial \bfz(\beta) }{ \partial \beta} \nonumber \\ & =
\sum_{k: \ 0 < z_k (\beta) < 1} -f_{r,k}^2 \nonumber \\ & = - \big| \{k \mid 1 \leq k \leq d,  0 <
z_k (\beta) < 1\} \big| \label{eq:derivative} \\
& < 0. \nonumber
\end{align}
proving that  $\bff_r ^T\bfz(\beta) $ is a decreasing function of $\beta$.
Therefore, by the mean value theorem, there exists a $\beta_{\textrm{opt}} \in [
  0 , \beta_{\textrm{max}}] $ such that
$\bff_{r}^T\bfz(\beta_{\textrm{opt}}) = r$.

First note that $\bfz(\beta_{\textrm{opt}})$ is feasible for Problem
PQP. We need only verify~(\ref{eq:constitParityConstraint}). Recalling that $r$ is
defined as $r = \lfloor \|
\Proj_{[0,1]^d}(\bfv)\|_1\rfloor_{\textrm{even}}$, we get the lower
bound:
\begin{align*}
\mathbf{1}^T \bfz(\beta_{\textrm{opt}}) \geq \bff_r^T\bfz(\beta_{\textrm{opt}})  = r.
\end{align*}
The components of $\bfz(\beta_{\textrm{opt}})$ are all less than one, so $\sum_{k=1}^{r+1}
z_k(\beta_{\textrm{opt}}) \leq r+1$. Combining this with the equality
$\bff_r^T \bfz(\beta_{\textrm{opt}})  = r$ tells us that
$\sum_{k=r+2}^d z_k(\beta_{\textrm{opt}}) 
\leq 1 $. We therefore find that $\mathbf{1}^T \bfz(\beta_{\textrm{opt}})$ is no more than
$r+2$.

To complete the proof, we need only find dual variables to certify the
optimality.  Setting $\lagxi$, $\lagzeta$, and $\lagtheta$ to zero, and $\lagmu$
and $\laggamma$ to the values required to satisfy~\eq{kkt-grad} provides the
necessary assignments to satisfy the KKT conditions.
\end{proof}

Lemma~\ref{lemma:optimality} thus certifies that all we need to do to
compute the projection is to compute the optimal $\beta$.  To do so,
we use the fact that the function $\bff_r^T \bfz(\beta)$ is a
piecewise linear function of $\beta$.  For a fixed $\beta$, define
the \emph{active set} to be the indices where $\bfz(\beta)$ is strictly
between $0$ and $1$ 
\begin{equation}
\mathcal{A}(\beta) := \{ k \mid 1\leq k \leq d, 0 < z_k(\beta) < 1 \}\,.
\label{eq.defActiveSet}
\end{equation}
Let the \emph{clipped set} be the indices where $\bfz(\beta)$ is
equal to $1$.
\begin{equation}
\mathcal{C}(\beta) := \{ k \mid 1\leq k \leq d,  z_k(\beta) = 1 \}\,.
\label{eq.defClippedSet}
\end{equation}
Let the \emph{zero set} be the indices where $\bfz(\beta)$ is equal to
zero
\begin{equation}
\mathcal{Z}(\beta) := \{ k \mid 1\leq k \leq d,  z_k(\beta) = 0 \}\,.
\label{eq.defZeroSet}
\end{equation}
Note that with these definitions, we have
\begin{align}
\bff_r^T \bfz(\beta) &= |\mathcal{C}(\beta)| + \sum_{j \in
  \mathcal{A}(\beta)} (z_j - \beta) \nonumber\\
&= |\mathcal{C}(\beta)| -\beta |\mathcal{A}(\beta)| + \sum_{j \in
  \mathcal{A}(\beta)} z_j \label{eq:find-beta}
\end{align}

Our algorithm simply increases $\beta$ until the active set changes,
keeping track of the sets $\mathcal{A}(\beta)$, $\mathcal{C}(\beta)$,
and $\mathcal{Z}(\beta)$.  We break the interval
$[0,\beta_{\mathrm{max}}]$ into the locations where the active set
changes, and compute the value of $\bff_r^T \bfz(\beta)$ at each of
these breakpoints until $\bff_r^T \bfz(\beta)<r$.  At this point, we
have located the appropriate active set for optimality and can find
$\beta_{\mathrm{opt}}$ by solving the linear equation~\eq{find-beta}.

The breakpoints themselves are easy to find: they are the values of
$\beta$ where an index is set equal to one or equal to zero.  First,
define the following sets
\begin{align*}
\mcE_1 & := \{ v_i - 1 \mid 1 \leq i \leq r+1 \},  \\
\mcL_1 & := \{v_i \mid 1 \leq i \leq r+1 \}, \\ 
\mcE_2 & := \{ - v_i \mid r+2  \leq i \leq d  \}, \\
\mcL_2 & := \{ - v_i +1  \mid r+2  \leq i \leq d \}. 
\end{align*}
The sets $\mcE_1$ and $\mcL_1$ concern the $r+1$ largest
components of $\bfv$; $\mcE_2$ and $\mcL_2$ the smallest
components. 
The set of possible breakpoints is
\begin{align*}
\mcB := \left\{  \left.\beta \in \mcE_1 \cup \mcE_2 \cup \mcL_1 \cup\mcL_2 \right| 0
\leq \beta \leq \beta_{\textrm{max}} \right\} \cup \{0,
\beta_{\textrm{max}} \}.
\end{align*}

\begin{figure}[h!]
\psfrag{&A}{$\bff_r^T \bfz(\beta)$}
\psfrag{&B}{$r$}
\psfrag{&C}{$0$}
\psfrag{&D}{$\beta_{i-1}$}
\psfrag{&E}{ $\beta_{\textrm{opt}}$ } 
\psfrag{&F}{$\beta_{i+1}$}

\psfrag{&G}{$\beta_{\textrm{max}}$}
\psfrag{&H}{$\beta$}
\psfrag{&I}{} 
\psfrag{&J}{}   \centering
  \includegraphics[width=3.5in]{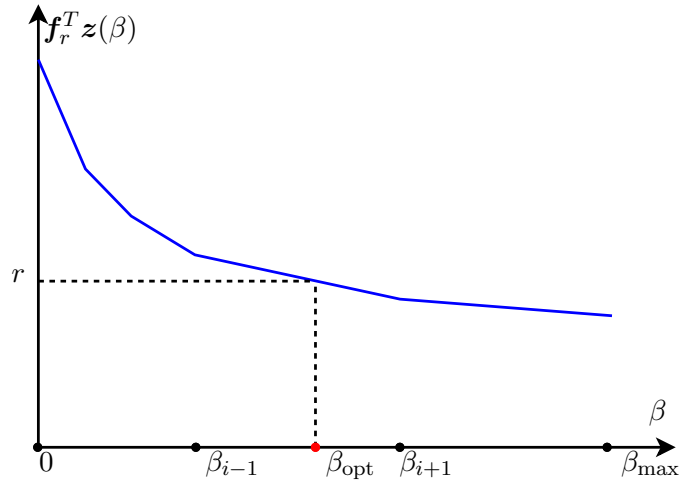}
\caption{Since there are a finite number of breakpoints (at most $2d+2$) and the function $ \bff_r^T\bfz(\beta)$ is linear between breakpoints, we can solve for $\beta_{\textrm{opt}}$ in linear time.  See (\ref{eq:fr-def}) and (\ref{def.zProj}) for definitions of $\bff_r$ and $\bfz(\beta)$ respectively.}
\label{fig:breakpoints}
\end{figure}

The following proposition reduces the search space by identifying unnecessary breakpoints.
\begin{proposition}
$\beta_{\textrm{opt}} \leq v_{r+1}$ and $\beta_{\textrm{opt}} < 1 - v_{r+2}$.
\end{proposition}
\begin{proof}
If $v_{r+1} > -v_{r+2}$, then $v_{r+1} > \frac{1}{2}(v_{r+1} - v_{r+2}) = \beta_{\textrm{max}}$. Therefore $\beta_{\textrm{opt}} < \beta_{\textrm{max}} < v_{r+1}$. If $v_{r+1} \leq -v_{r+2}$, then for all $\hat{\beta} \geq v_{r+1}$, we have $\bff_r^Tz(\hat{\beta})\leq r$. This means that for all $\hat{\beta} > v_{r+1}$, $\bff_r^Tz(\hat{\beta})< r$ and hence $\hat{\beta}$ cannot be $\beta_{\textrm{opt}}$. 

In addition, $\beta_{\textrm{opt}} <1 - v_{r+2}$. This is because for all $\hat{\beta} \geq 1- r_{r+2}$, we have $\|z(\hat{\beta})\|_1 \geq r+2$. Therefore $z(\hat{\beta})$ cannot be the projection.
\end{proof}
The proposition above indicates that the points in $\mcL_1$ and $\mcL_2$ are not necessary when identifying $\beta_{\textrm{opt}}$. The reason is that the only possible solution in these two sets is $\beta_{\textrm{opt}} = v_{r+1}\in\mcL_1$. However, since $v_{r+1}$ is the smallest breakpoint, we can directly solve for $\beta_{\textrm{opt}}$ by inspecting the second smallest breakpoint and use~\eq{find-beta}. This procedure is captured in Algorithm~\ref{Algorithm:PPd}. We use $\mathcal{B} = \left\{ \beta \in \cup_{j=1}^2 \mathcal{E}_j   \mid 0 \leq \beta \leq \beta_{\textrm{max}}  \right\} $ as the set of breakpoints in Algorithm~\ref{Algorithm:PPd}. This set contains at most $d$ points. 

To summarize, our Algorithm~\ref{Algorithm:PPd} sorts the input
vector, computes the set of breakpoints, and then marches through the
breakpoints until it finds a value of $\beta_i\in\mathcal{B}$ with
$\bff_r^T \bfz(\beta_i)\leq r$.  Since we will also have $\bff_r^T
\bfz(\beta_{i-1})> r$, the optimal $\beta$ will lie in
$[\beta_{i-1},\beta_i]$ and can be found by solving~\eq{find-beta}.
In the algorithm box for Algorithm~\ref{Algorithm:PPd}, $b$ is the
largest and $a$ is the smallest index in the active set.  We use $V$
to denote the sum of the elements in the active set and $\Lambda$ the
total sum of the vector at the current break point.  Some of the
awkward {\bf if} statements in the main {\bf for} loop take care of
the cases when the input vector has many repeated entries.

Algorithm~\ref{Algorithm:PPd} requires one sort (sorting the input vector), 
and inspections of at most $d$
breakpoints.  Thus, the total complexity of the algorithm is linear
plus the time for the sort, which is $O(d\log d)$. We make two remarks on 
the computational complexity of Algorithm~\ref{Algorithm:PPd}. First, $d$ 
is small for LDPC codes. Thus the asymptotic complexity is less important. 
The complexity from the sorting operation is negligible. 
Second, we can consider Algorithm~\ref{Algorithm:PPd} as the check node 
operation. Compared with an exact sum-product BP check node update, where hyperbolic 
tangents and logarithms are used, our algorithm uses only basic arithmetic 
operations and thus is more hardware friendly. We demonstrate in numerical
results that the actual run time for each iteration of ADMM decoding is similar to
that for BP decoding.

\begin{algorithm}
\caption{Given $\bfu \in \mathbb{R}^d$ determine its projection on
  $\PP_d$, $\bfz^*$}
\label{Algorithm:PPd}
\begin{algorithmic}[1]

\STATE Permute $\bfu$ to produce a vector $\bfv$ whose components
are sorted in decreasing order, i.e., $v_1 \geq v_2 \geq \ldots \geq
v_d$.  Let $\bfQ$ be the corresponding permutation matrix, i.e., $\bfv = \bfQ \bfu$.
\STATE Compute $ \bfzHat \leftarrow \Proj_{[0,1]^d} (\bfv) $.
\STATE Assign $r = \lfloor \|\bfzHat\|_1 \rfloor_{\textrm{even}} $ 
\IF{$r = d$}
\STATE \textbf{return} $\bfz^*=\bfzHat$.
\ENDIF

\IF{$r\leq d-2$}
\STATE $\beta_{\textrm{max}} = \frac{1}{2} [v_{r+1} - v_{r+2}]$.
\ELSE 
\STATE $\beta_{\textrm{max}} = v_{r+1}$.
\ENDIF

\STATE Define $\bff_r$ as in~\eq{fr-def}.

\IF{ $ \bff_{r} ^T\bfzHat  \leq r$ }
\STATE \textbf{return} $\bfz^*=\bfzHat$.
\ENDIF

\STATE \begin{tabbing} Assign \= $ \mathcal{E}_1 = \{ v_i - 1 \mid 1
  \leq i \leq r+1 \}$, \\ \> $ \mathcal{E}_2 = \{ - v_i \mid r+2 \leq i \leq d \}$.
 \end{tabbing}
\STATE Construct the set of breakpoints \\
$\mathcal{B} := \left\{ \beta \in \mathcal{E}_1\cup\mathcal{E}_2 \mid 0 \leq \beta \leq \beta_{\textrm{max}}  \right\} $ \\
by merging the sorted points in $\mathcal{E}_1$ and $\mathcal{E}_2$ so that entries in $\mcB$
satisfy $\beta_1 \leq \beta_2 \leq \ldots \leq \beta_{|\mathcal{B}|}$.
\STATE Initialize $a$ as the smallest index such that $0 < \hat{z}_a < 1$. 
\STATE Initialize $b$ as the largest index such that $0 < \hat{z}_b < 1$.
\STATE Initialize sum $V = \bff_r^T \bfzHat$. \\

\FOR{ $i = 1$ to $|\mathcal{B}|$ }

\STATE Set $\beta_0 \leftarrow \beta_i$.
\IF{ $ \beta_i \in \mathcal{E}_1 $}
\STATE Update $ a \leftarrow a-1$. 
\STATE Update $V \leftarrow V + v_a $.  
\ELSE
\STATE  Update $b \leftarrow b+1 $
\STATE Update $V \leftarrow V - v_b$. 
\ENDIF 

\IF{$i<|\mcB|$ and $\beta_i \neq \beta_{i+1}$}
   \STATE \label{Step:using-V} $\Lambda \leftarrow (a-1) + V - \beta_0 ( b-a +1 ) $ 
   \STATE {\bf if} { $ \Lambda \leq r$ } {\bf then break}
\ELSIF {$i=|\mcB|$}
   \STATE $\Lambda \leftarrow (a-1) + V - \beta_0 ( b-a +1 ) $
\ENDIF

\ENDFOR

\IF{$\Lambda>r$}
\STATE \label{Step:compute-beta-opt} Compute $\beta_{\textrm{opt}} \leftarrow
\beta_0 - \frac{ r -
  \Lambda } { b-a +1 }$.  
\ELSE
\STATE $\beta_0\leftarrow \beta_{i-1}$
\STATE $a \leftarrow |\{j~|~v_j-\beta_0>1\}|$
\STATE $b \leftarrow r+2 + |\{j~|~v_j+\beta_0 0\}|$
\STATE $V\leftarrow \sum_{j=a-1}^{r+1} v_j - \sum_{j=r+2}^{b+1} v_j$
\STATE $\beta_{\mathrm{opt}} \leftarrow  \frac{V+b-r}{b-a+1}$
\ENDIF
\STATE \textbf{return} $\bfz^*=\bfQ^T \Pi_{[0,1]^d}(\bfv - \beta_{\textrm{opt}} \bff_r)$. 
\end{algorithmic}
\end{algorithm}
\section{Numerical results and implementation}
\label{sec.simulations}
In this section, we present simulation results for the ADMM decoder
and discuss various aspects of our implementation.  In
Section~\ref{sec.simRes} we present word-error-rate (WER) results for
two particular LDPC codes as well as for an ensemble of random
$(3,6)$-regular LDPC codes.  We note that numerical results for the
$(155,64)$ Tanner code~\cite{tanner:01} is reported in our previous
work~\cite{barmanEtAl:allerton11}.  In Section~\ref{sec.paramChoices}
we discuss how the various parameters choices in ADMM affect decoding
performance, as measured by error rate and by decoding time.

\subsection{Performance comparisons between ADMM and BP decoding}
\label{sec.simRes}

In this section, we present simulation results of ADMM decoding and
compare to sum-product BP decoding. The parameters used for ADMM are
as follows: (i) error tolerance $\epsilon = 10^{-5}$, (ii) penalty
$\mu = 3$, (iii) maximum number of iterations $t_{\max} = 1000$ and
(iv) over-relaxation parameter (cf. Sec.\ref{sec.paramChoices}) $\rho
= 1.9$. The maximum number of iterations for BP decoding is also
$1000$. We discuss parameter choices in detail in
Sec.~\ref{sec.paramChoices}.

We first present results for two particular codes over the additive
white Gaussian noise (AWGN) channel with binary inputs.  The first
code is the $[2640,1320]$ rate-$0.5$, $(3,6)$-regular Margulis LDPC
code~\cite{ryanLin:2009channel}.  The second is a $[1057,813]$
rate-$0.77$, $(3,13)$-regular LDPC code obtained
from~\cite{mackayWeb}.  This code is also studied by Yedidia \emph{et
  al.}~\cite{yedidiaWangDraper:IT11}.  We choose both codes as they
have been chosen in the past to study error floor performance.  Then
we present results for an ensemble of $100$ randomly generated
$(3,6)$-regular LDPC codes of length $1002$ similar
to~\cite{burshtein:IT11}.  We simulate the binary symmetric channel
(BSC) and compare ADMM with BP decoding in three aspects: error rate,
number of iterations and execution time.\footnote{We note that results
  for ADMM decoding of the $(155,64)$ Tanner code are given
  in~\cite{barmanEtAl:allerton11}, results that match those given
  in~\cite{draperYedidiaWang:ISIT07}.}

In Fig.~\ref{fig:AWGN_Margulis} we plot the WER performance of the
Margulis code for the ADMM decoder and various implementations of
sum-product BP decoding. As mentioned, this code has been extensively
studied in the literature due to its error floor behavior (see,
e.g.,~\cite{mackayPostol:03, ryanLin:2009channel,
  butlerSiegel:allerton11}). Recently it has been
noted~\cite{butlerSiegel:allerton11,butlerSiegel:preprint12} that the
previously observed error floor of this code is, at least partially, a
result of saturation in the message LLRs passed by the BP decoder.
This issue of implementation can be greatly mitigated by improving the
way large LLRs are handled.  Thus, alongside these previous results we
plot results of our own implementation of ``non-saturating''
sum-product BP, which follows the implementation
of~\cite{butlerSiegel:allerton11,butlerSiegel:preprint12}, and which
matches the results reported therein.  In our simulations of the ADMM
decoder, we collect more than 200 errors for all data points other than
the 2 highest SNRs ($2.8$ dB and $3$ dB), for which we collected $130$
and $32$ errors respectively. For non-saturating BP decoding, we
collect $43$ and $13$ respective errors at SNR $= 2.6,2.7$ dB.

\begin{figure}
\psfrag{&A}{\tiny{ADMM}}
\psfrag{&B}{\tiny{BP decoding (Ryan and Lin)}}
\psfrag{&C}{\tiny{BP decoding (Mackay)}}
\psfrag{&D}{\tiny{Non-saturating BP}}
\psfrag{&E}{\hspace*{-1em}\tiny{$E_b/N_0$ (dB)}}
\psfrag{&F}{\hspace*{-3em}\tiny{word-error-rate (WER)}}
    \begin{center}
    \includegraphics[width=9.25cm]{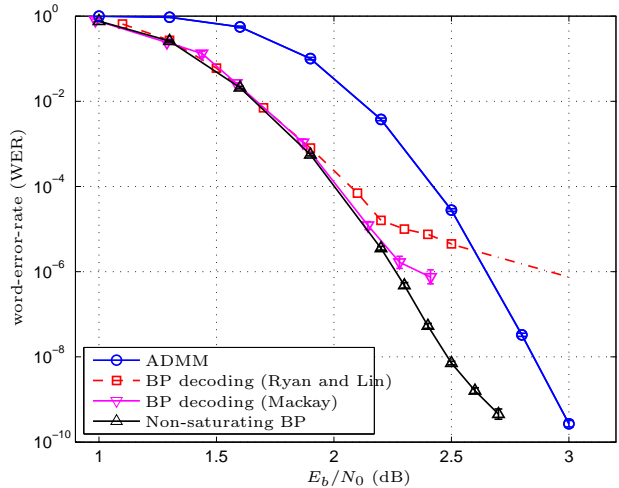}
    \end{center}
    \caption{Word error rate (WER) of the $[2540,1320]$ ``Margulis''
      LDPC code used on the AWGN channel plotted as a function of
      signal-to-noise ratio (SNR).  The WER performance of ADMM is
      compared to that of non-saturating sum-product BP, as well as to
      results for (saturating) sum-product BP from Ryan and
      Lin~\cite{ryanLin:2009channel} and from MacKay and
      Postol~\cite{mackayPostol:03}.}
      \label{fig:AWGN_Margulis}
\end{figure}

The first aspect to note is that while the LP decoder has a waterfall,
the waterfall initiates at a slightly higher SNR (about $0.4$ dB
higher in this example) than that of sum-product BP.  This observation
is consistent with earlier simulations of LP decoding for long block
lengths, e.g., those presented in~\cite{yedidiaWangDraper:IT11,
  wangYedidiaDraper:ISIT09}.  It is worth mentioning that expressing
BP decoding as optimization over the Bethe free energy is introduced
in~\cite{yedidiaFreemanWeiss:IT05}.  Further studies such
as~\cite{vontobel:10preprint} show that BP and LP decoding, when
expressed using the Bethe free energy, are different in the objective
function. Therefore, one should not expect identical performance, as
the simulations demonstrate.

The second aspect to note is that, as in the prior work, we do not
observe an error floor in ADMM decoding at WERs above $10^{-10}$.
When decoding of this code using the non-saturating version of
sum-product, we observe a weak error floor at WERs near $10^{-9}$, in
which regime the waterfall of ADMM is continuing to steepen.  In this
regime we found that the non-saturating BP decoder is oscillating, as
discussed in~\cite{zhangEtAl:tcomm09}
\cite{ruozziThalerTatikonda:allerton09}.  We note that we have not
simulated WERs at $10^{-10}$ or lower due to the limitation of our
computational resources.  It would be extremely interesting to see the
performance of ADMM decoding at WERs lower than $10^{-10}$.

Figure~\ref{fig:AWGN_1057} presents simulation results for the
rate-$0.77$ length-$1057$ code.  In this simulation, all data points
are based on more than $200$ errors except for the ADMM data at SNR =
$5$ dB, where $29$ errors are observed. In addition we plot an
estimated lower bound on maximum likelihood (ML) decoding
performances.  The lower bound is estimated in the following way. In
the ADMM decoding simulations we round any non-integer solution
obtained from the ADMM decoder to produce a codeword estimate. If the
decoder produces a decoding error, i.e., if the estimate does not
match the transmitted codeword, we check if the estimate is a valid
codeword.  If the estimate satisfies all the parity checks (and is
therefore a codeword) we also compare the probability of the estimate
given the channel observations with the that of the transmitted
codeword given the channel observations.  If the probability of
estimate is greater than that of the transmitted codeword we know that
an ML decoder would also be in error.  All other events are counted as
ML successes (hence the estimated {\em lower} bound on ML
performance).  Similar to the Margulis code,
Fig.~\ref{fig:AWGN_1057} shows that for this code the ADMM decoder
displays no signs of an error floor, while the BP decoder does.
Further, ADMM is approaching the ML error lower bound at high SNRs.

\begin{figure}
\psfrag{&A}{\tiny{ADMM}}
\psfrag{&B}{\tiny{Non-saturating BP}}
\psfrag{&C}{\tiny{ML lower bound}}
\psfrag{&D}{\hspace*{-1em}\tiny{$E_b/N_0$ (dB)}}
\psfrag{&E}{\hspace*{-3em}\tiny{word-error-rate (WER)}}
    \begin{center}
    \includegraphics[width=9.25cm]{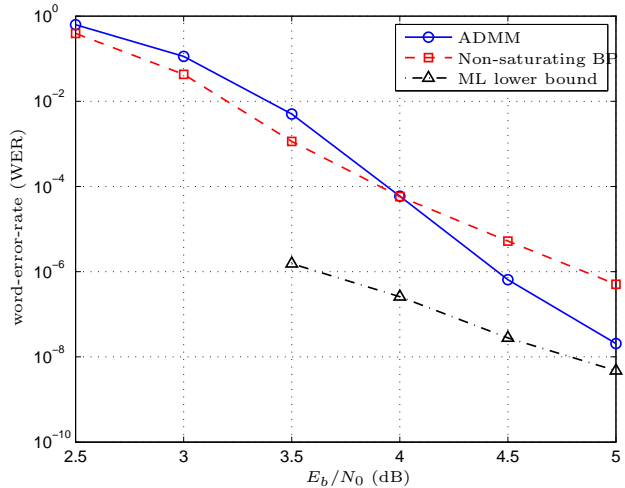}
    \end{center}
    \caption{Word error rate (WER) of the $[1057,813]$ LDPC code used
      on the AWGN channel plotted as a function of signal-to-noise
      ratio (SNR).  The WER performance of ADMM is compared to that of
      non-saturating sum-product BP, as well as to an estimated
      lower-bound on ML decoding.}
      \label{fig:AWGN_1057}
\end{figure}

In Fig.~\ref{fig:RandomCodeErrorRate},~\ref{fig:RandomCode_BSC_iter}
and~\ref{fig:RandomCode_BSC_time}, we present comparisons between ADMM
decoding and BP decoding using an ensemble of $100$ randomly generated
$(3,6)$-regular LDPC codes of length $1002$. We eliminated codes that
have parallel edges, thus all codes have girth of at least four.
However, cycles of length four or greater are not eliminated.  We will
use this ensemble to understand the error performance and the
computational performance of LP and of BP decoding.  For this study we
simulate the BSC in order to match the settings used
in~\cite{burshtein:IT11}.  All data points presented are averaged
across the 100 codes in the ensemble. For each code, we collect more
than $5$ word-errors. 

In Fig.~\ref{fig:RandomCodeErrorRate} we plot the average
word-error-rate (WER) and bit-error-rate (BER) observed for both BP
and ADMM decoding.  We observe similar comparisons between ADMM and BP
decoding found in previous examples.  In particular, note the
error floor flare observable in BP at cross-over probabilities of
about $0.045$ and below.  No such flare is evident in ADMM.

In Fig.~\ref{fig:RandomCode_BSC_iter} we plot a comparison of the
iteration requirements of ADMM and BP decoding for the same ensemble
of codes. We plot three curves for each decoder: the average number of
iterations required to decode, the average number of iterations
required to decode when decoding is correct, and the average number
required when decoding is erroneous.  We observe that ADMM decoding
needs more iterations to decode than BP does. However, the gap between
the decoders is roughly constant (on this log scale) meaning the ratio
of iterations required is roughly constant.  Thus, the trend for
increased iterations at higher crossovers is the same for both
decoders.  Further, both decoder reach the maximum number of allowable
iterations when errors occur. An important observations is that
although we allow up to $1000$ iterations in our simulations, the
average number of iterations required by ADMM for correct decoding
events is quite small at \emph{all} SNRs.  This means that ADMM
converges quickly to a correct codeword, but more slowly to a
pseudocodeword.  We discuss further the effect of choice of the
maximum number of iterations in Sec.~\ref{sec.paramChoices}.

In Fig.~\ref{fig:RandomCode_BSC_time} we plot the time comparisons
between ADMM and BP decoding using the same methodology. For this
figure we plot results for the {\em saturating} version of BP where we
have tried to optimized our implementations.  This decoder executes
{\em much} more quickly than our implementation of non-saturating BP.
Both decoders are simulated on the same CPU configurations.  We make
two observations.  First, when measured in terms of execution time,
the computational complexity of ADMM and BP are similar.  This
observation holds for all crossover probabilities simulated. Second,
ADMM decoding is faster than BP when decoding is correct. Combining
these results with those on iteration count from
Fig.~\ref{fig:RandomCode_BSC_iter} we conclude that the execution time
for each iteration of ADMM is shorter than for BP.

\begin{figure}
\psfrag{&A}{\tiny{ADMM, WER}}
\psfrag{&B}{\tiny{BP, WER}}
\psfrag{&C}{\tiny{ADMM, BER}}
\psfrag{&D}{\tiny{BP, BER}}
\psfrag{&E}{\hspace*{-2em}\tiny{crossover probability}}
\psfrag{&F}{\hspace*{-1em}\tiny{error rate}}
    \begin{center}
    \includegraphics[width=9.25cm]{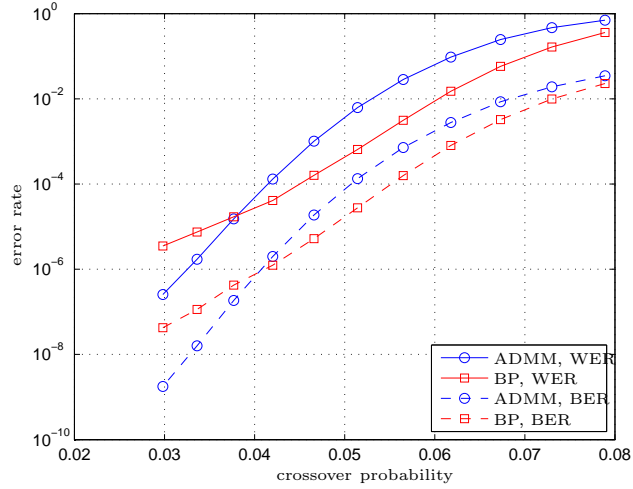}
    \end{center}
    \caption{Word error rate (WER) and bit-error-rate (BER) of the $(3,6)$-regular random
      LDPC code used on the BSC plotted as a function of
      crossover probability.  The error rate performance of ADMM is
      compared to that of saturating sum-product BP. Results are averaged
      over 100 randomly generated codes.}
      \label{fig:RandomCodeErrorRate}
\end{figure}
\begin{figure}
\psfrag{&A}{\hspace*{-2em}\tiny{crossover probability}}
\psfrag{&B}{\hspace*{-2em}\tiny{\# of iterations}}
\psfrag{&C}{\tiny{ADMM, erroneous}}
\psfrag{&D}{\tiny{BP, erroneous}}
\psfrag{&E}{\tiny{ADMM, average}}
\psfrag{&F}{\tiny{BP, average}}
\psfrag{&G}{\tiny{ADMM, correct}}
\psfrag{&H}{\tiny{BP, correct}}
    \begin{center}
    \includegraphics[width=9.25cm]{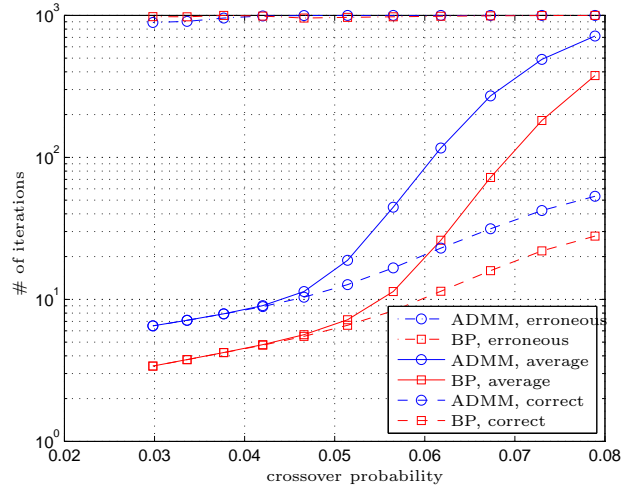}
    \end{center}
    \caption{Number of iterations of the $(3,6)$-regular random
      LDPC code used on the BSC plotted as a function of
      crossover probability.  The number of iterations of ADMM is
      compared to that of saturating sum-product BP. Results are averaged
      over 100 randomly generated codes.}
      \label{fig:RandomCode_BSC_iter}
\end{figure}
\begin{figure}
\psfrag{&A}{\hspace*{-2em}\tiny{crossover probability}}
\psfrag{&B}{\hspace*{-2em}\tiny{execution time (sec)}}
\psfrag{&C}{\tiny{ADMM, erroneous}}
\psfrag{&D}{\tiny{BP, erroneous}}
\psfrag{&E}{\tiny{ADMM, average}}
\psfrag{&F}{\tiny{BP, average}}
\psfrag{&G}{\tiny{ADMM, correct}}
\psfrag{&H}{\tiny{BP, correct}}
    \begin{center}
    \includegraphics[width=9.25cm]{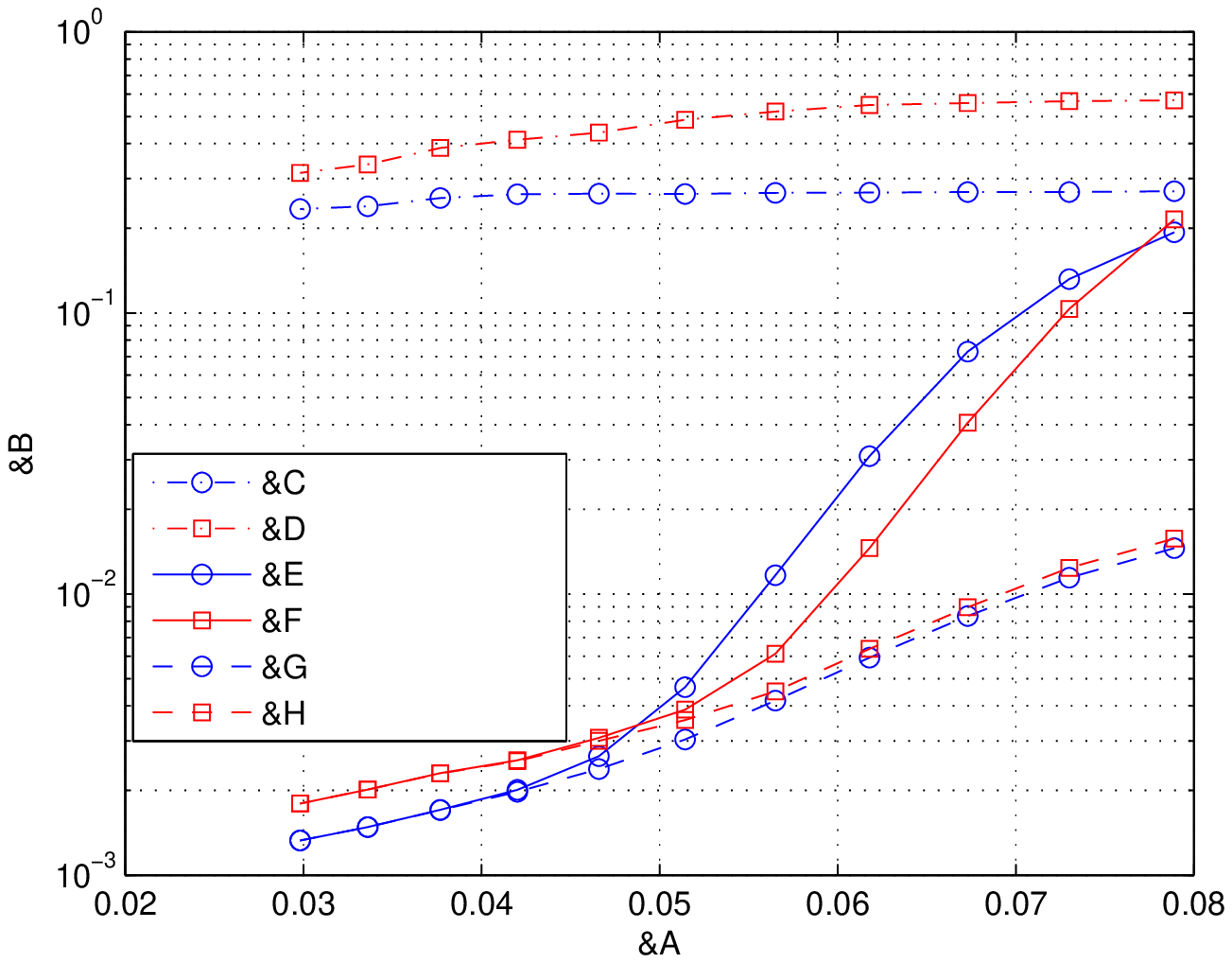}
    \end{center}
    \caption{Execution time of the $(3,6)$-regular random
      LDPC code used on the BSC plotted as a function of
      crossover probability.  The execution time of ADMM is
      compared to that of saturating sum-product BP. Results are averaged
      over 100 randomly generated codes.}
      \label{fig:RandomCode_BSC_time}
\end{figure}

Given the importance of error floor effects in high reliability
applications, and the outcomes of our simulations, we now make some observations.  One point
demonstrated by these experiments, in particular by the simulation of
the Margulis code (and also argued
in~\cite{butlerSiegel:allerton11,butlerSiegel:preprint12}) is that
numerical precision effects can dramatically affect code performance
in the high SNR regime.  From a practical point of view, a real-world
implementation would use fixed precision arithmetic.  Thus,
understanding the behavior of ADMM decoding under finite precision is
extremely important.

A second point made by comparing these codes is that the performance
of an algorithm, e.g., non-saturating BP, can vary dramatically from
code to code (Margulis vs.\ 1057), and the performance of a code can
vary dramatically from algorithm to algorithm (BP vs.\ ADMM).  For
each algorithm we might think about three types of
codes~\cite{yedidia:corresp12}.  The first (type-A) would consist of
codes that do not have any trapping sets, i.e., do not display an
error floor, even for low-precision implementations.  The second
(type-B) would consist of codes whose behavior changes with precision
(e.g., the Margulis code).  The final (type-C) would consist of codes
that have trapping sets even under infinite precision (the
length-$1057$ code may belong to this set).  Under this taxonomy there
are two natural strategies to pursue.  The first is to design codes
that fall in the first class.  This is the approach taken in,
e.g.,~\cite{osullivan:IT04} \cite{fossorier:IT04} \cite{huEtAl:IT05}
\cite{milenkovicEtAl:IT06} \cite{wangDraperYedidia:preprint11}, where
codes of large-girth are sought.  The second is to design improved
algorithms that enlarge the set of codes that fall into the first
class.  This is the approach taken in this paper.  Some advantageous
numerical properties of ADMM are as follows:  First, ADMM has rigorous
convergence guarantees~\cite{boydEtAl:FnT10}. Second, ADMM has
historically be observed to be quite robust to parameter choices and
precision settings~\cite{boydEtAl:FnT10}. This robustness will be
further demonstrated in Sec.~\ref{sec.paramChoices}.  Third, the
``messages'' passed in ADMM (the replica values) are inherently
bounded to the unit interval (since the parity polytope is contained
within the unit hypercube). Due to these numerical properties, we
expect that the ADMM decoder will be a strong competitor to BP in
applications that demand ultra-high reliabilities.

\subsection{Parameter choices}
\label{sec.paramChoices}

In the ADMM decoding algorithm there are a number of parameters that
need to be set.  The first is the stopping tolerance, $\epsilon$, the
second is the penalty parameter, $\mu$, and the third is the maximum
allowable number of iterations, $t_{\max}$.  In our experiments we
explored the sensitivity of algorithm behavior, in particular
word-error-rate and execution-time statistics, as a function of the
settings of these parameters.  In this section we present results that
summarize what we learned.  We report results for the Margulis LDPC
code as used in the AWGN channel.  This is consistent with the
simulations presented in the last subsection.

We first explore the effects of the choice of $\epsilon$ and $\mu$ on
error rate.  We comment that as long as $t_{\max} > 300$ the choice of
$t_{\max}$ does not significantly affect the WER. This effect is also
evidenced in Fig.~\ref{fig:RandomCode_BSC_iter} where the average
number of iterations required in correct decoding events is seen to be
small.  In Fig.~\ref{fig:Margulis_Choose_etol} we plot WER as a
function of the number of bits of stopping tolerance, i.e., $- \log_2
(\epsilon)$. In Fig.~\ref{fig:Margulis_Choose_mu_WER} we plot WER as a
function of $\mu$.  Each data point is based on more than $200$
decoding errors.

\begin{figure}
\psfrag{&A}{\tiny{$\mu = 1$}}
\psfrag{&B}{\tiny{$\mu = 5$}}
\psfrag{&C}{\tiny{$\mu = 7$}}
\psfrag{&D}{\hspace*{-6em} \tiny{$-\log_2$(ending tolerance) $= -\log_2(\epsilon)$}}
\psfrag{&E}{\hspace*{-3em}\tiny{word-error-rate (WER)}}
    \begin{center}
    \includegraphics[width=9.25cm]{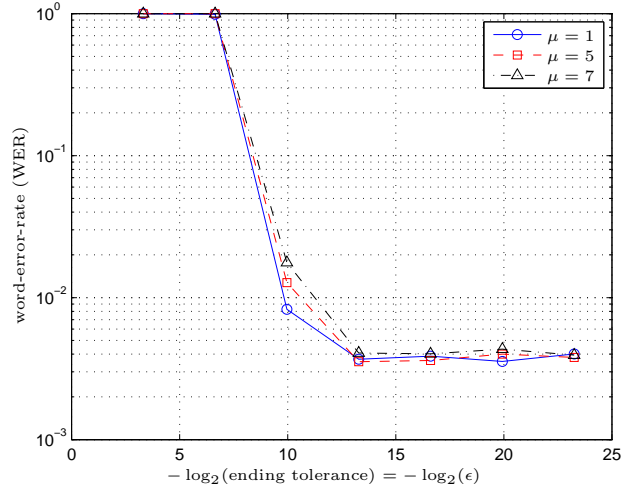}
    \end{center}
    \caption{The effect of the error tolerance $\epsilon$ on word
      error rate (WER). The WER of the Margulis LDPC code for the AWGN
      channel plotted as a function of error tolerance $\epsilon$ for
      three difference penalty parameters $\mu$. The SNR simulated is
      $2.2$ dB. The maximum number of iterations $t_{max}$ is set
      equal to $1000$.}
      \label{fig:Margulis_Choose_etol}
\end{figure}

\begin{figure}
\psfrag{&A}{\tiny{$E_b/N_0=1.9$ dB}}
\psfrag{&B}{\tiny{$E_b/N_0=2.2$ dB}}
\psfrag{&C}{\tiny{$E_b/N_0=2.5$ dB}}
\psfrag{&D}{\hspace*{-3em}\tiny{penalty parameter, $\mu$}}
\psfrag{&E}{\hspace*{-3em}\tiny{word-error-rate (WER)}}
    \begin{center}
    \includegraphics[width=9.25cm]{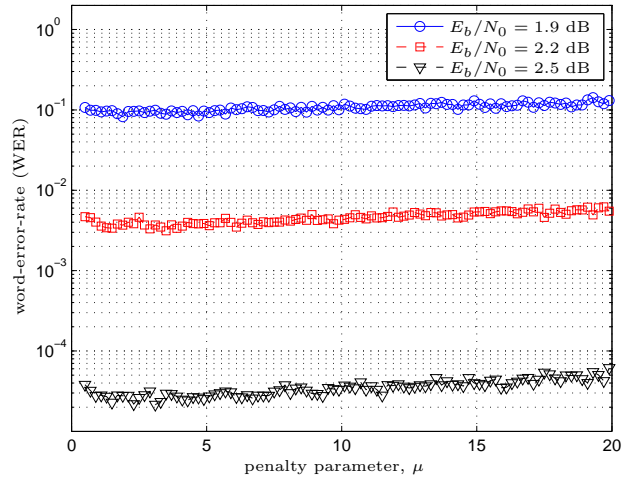}
    \end{center}
    \caption{The effect of the penalty parameter $\mu$ on word error
      rate (WER).  The WER of the Margulis LDPC code for the AWGN
      channel plotted as a function of penalty parameter $\mu$. Error
      tolerance $\epsilon = 10^{-5}$, and maximum number of iterations
      $t_{max} = 1000$.}
      \label{fig:Margulis_Choose_mu_WER}
\end{figure}

From these two figures we conclude that the error performance of the
ADMM decoder depends only weakly on the settings of these two
parameters.  A sufficiently large $\epsilon$ parameter and a moderate
$\mu$ parameter are good enough to achieve the desired error rate.
For instance $\epsilon \geq 10^{-4}$ and $\mu \geq 1$ should do.  This
means that the design engineer has great latitude in the choice of
these parameters and can make, e.g., hardware-compatible choices.
Furthermore, the results on ending tolerance give hints as to the
needed precision of the algorithm.  If algorithmic precision is on the
order of the needed ending tolerance we expect to observe similar
error rates.

We next study the effect of parameter section on average decoding
time.  All time statistics were collected on a 3GHz Intel(R) Core(TM)
2 CPU.  In Fig.~\ref{fig:Margulis_Choose_mu_time} we plot average
decoding time as a function of $\mu$ for three SNRs.  For all three
the ending tolerance is fixed at $\epsilon = 10^{-5}$.  We see some
variability in average decoding time as a function of the choice of
$\mu$. Recalling from Fig.~\ref{fig:Margulis_Choose_mu_WER} that we
should choose $\mu\in [1,10]$ for good WER performance, we conclude
that $\mu \in [2,5]$ is a good choice in term of both error- and
time-performance.

\begin{figure}
\psfrag{&A}{\tiny{$E_b/N_0 = 1.9$ dB}}
\psfrag{&B}{\tiny{$E_b/N_0 = 2.2$ dB}}
\psfrag{&C}{\tiny{$E_b/N_0 = 2.5$ dB}}
\psfrag{&D}{\hspace*{-3em}\tiny{penalty parameter, $\mu$}}
\psfrag{&E}{\hspace*{-3em}\tiny{execution time (sec)}}
    \begin{center}
    \includegraphics[width=9.25cm]{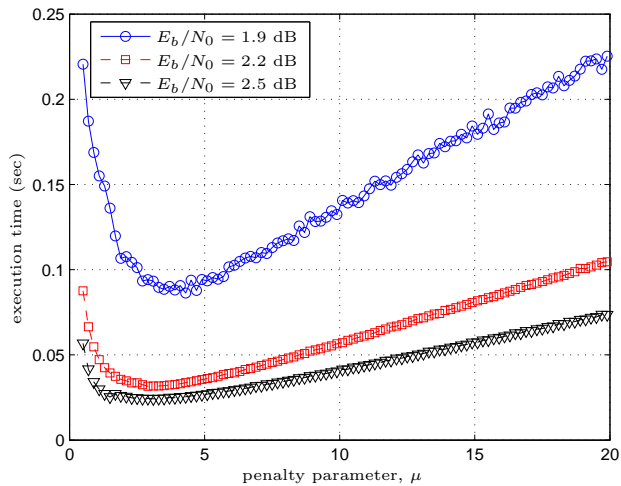}
    \end{center}
    \caption{The effect of the penalty parameter $\mu$ on execution
      time. Average execution time (in seconds) of ADMM decoding the
      Margulis code simulated over the AWGN channel plotted as a
      function of penalty parameter $\mu$ for three distinct SNRs.}
      \label{fig:Margulis_Choose_mu_time}
\end{figure}

\paragraph*{Over-relaxation}  A significant improvement in average decoding time results from
implementing an ``over-relaxed'' version of ADMM.  Over-relaxed ADMM
is discussed in~\cite[Section~3.4.3]{boydEtAl:FnT10} as a method for
improving convergence speed while retaining convergence guarantees.

The over-relaxation parameter $\overrel$ must be in the range $1 \leq
\overrel < 2$.  If $\overrel \geq 2$ convergence guarantees are lost.
We did simulated $\overrel > 2$ and observed an increase in average
decoding time.  In Fig.~\ref{fig:AWGN_Margulis_over_rlx} we plot the
effect on average decoding time of over-relaxed versions of the ADMM
decoder for $1 \leq \overrel \leq 1.9$.  These plots are for the
Margulis code simulated over the AWGN channel at an SNR of $2.8$
dB. We observe that the average decoding time drops by a factor of
about $50\%$ over the range of $\overrel$.  The improvement is roughly
constant across the set of penalty parameters studied: $\mu \in \{1,
3, 5, 7\}$.  By choosing the over-relaxation parameter $\overrel =
1.9$ we can double decoding efficiency without degradation in
error-rate.

While we did not use
over-relaxation in the experiments on parameter choices reported in
Figures~\ref{fig:Margulis_Choose_etol}
through~\ref{fig:Margulis_Choose_mu_time}, we would encourage
interested readers to explore proper settings of $\overrel$ in their
implementations.

\begin{figure}
\psfrag{&A}{\tiny{$\mu = 1$}}
\psfrag{&B}{\tiny{$\mu = 3$}}
\psfrag{&C}{\tiny{$\mu = 5$}}
\psfrag{&D}{\tiny{$\mu = 7$}}
\psfrag{&E}{\hspace*{-4em}\tiny{over-relaxation parameter, $\overrel$}}
\psfrag{&F}{\hspace*{-3em}\tiny{execution time (sec)}}
    \begin{center}
    \includegraphics[width=9.25cm]{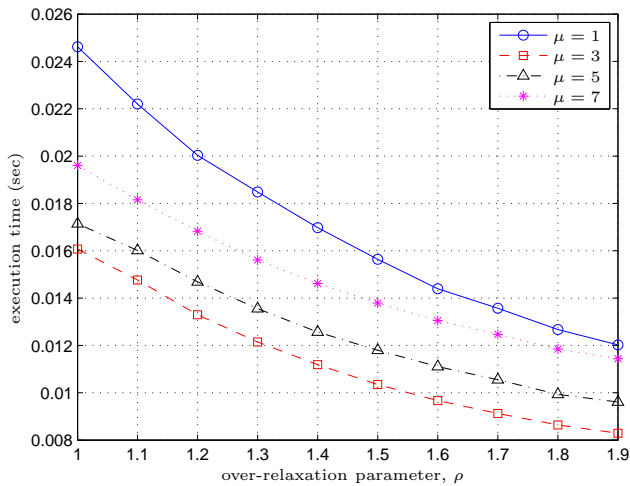}
    \end{center}
    \caption{The effect of the over-relaxation parameter $\rho$ on
      execution time.  Average execution time for ADMM decoding the
      Margulis code simulated over the AWGN channel at $E_b/N_0 = 2.8$
      dB.  Execution time (in seconds) is plotted as a function of
      over-relaxation parameter $\overrel$ for four different penalty
      parameters $\mu \in \{1, 3, 5, 7\}$.}
      \label{fig:AWGN_Margulis_over_rlx}
\end{figure}

\section{Conclusion}
\label{sec.conclusion}

In this paper we apply the ADMM template to the LP decoding problem
introduced in~\cite{feldmanEtAl:IT05}.  A main technical hurdle was
the development of an efficient method of projecting a point onto the
parity polytope.  We accomplished this in two steps.  We first
introduced a new ``two-slice'' representation of points in the parity
polytope.  We then used the representation to attain the projection
via an efficient waterfilling-type algorithm.  

We demonstrate the effectiveness of our decoding technique on the
rate-$0.5$ $[2640,1320]$ ``Margulis'' LDPC code, the rate-$0.77$
$[1057,813]$ LDPC code studied in~\cite{yedidiaWangDraper:IT11}, and
an ensemble of randomly generated $(3,6)$-regular LDPC codes.  We find
that while the decoding behaviors of LP and sum-product BP are similar
in many aspects there are also significant differences.  On one hand,
the waterfall of LP decoding initiates at slightly higher SNR than
that of sum-product BP decoding.  On the other, LP decoding does not
seem to have an error floor.  LP decoding, when implemented in a
distributed, scalable manner using ADMM, is a strong competitor to BP
in the high-SNR regime.  It allows LP decoding to be implemented as a
message-passing algorithm with a simple update schedule that can be
applied to long block-length codes with execution times similar to BP.

An immediate question is how to close the SNR gap that exists between
LP and BP decoding at low SNRs.  In recent work~\cite{liuEtAl:12}, the
authors introduce a penalized LP that increases the relative cost of
pseudocodewords vis-a-vis all-integer solutions.  While the resulting
optimization is (slightly) non-convex the ADMM framework introduced
herein can be applied with negligible (or no) increase in
computational complexity.  This slight modification closes the SNR gap
between BP and LP decoding while retaining the high-SNR behavior of LP
decoding.  Fully understanding the performance of the modified LP
decoders introduced in~\cite{liuEtAl:12} is an important future
direction.  Other interesting directions include application of the
decoder to other classes of codes, study of finite-precision effects,
and generalization of the two-slice representation to a larger class
of easy-to-project upon polytopes.

\section*{Acknowledgments}

The authors would like to thank Matthew Anderson, Brian Butler, Eric
Bach, Alex Dimakis, Paul Siegel, Emre Telatar, Yige Wang, Jonathan
Yedidia, and Dalibor Zelen\'{y} for useful discussions and references.
The authors would also like to note that some of the simulation
results presented in this research would not have been possible without
the resources and the computing assistance of the University of
Wisconsin (UW), Madison, Center For High Throughput Computing (CHTC)
in the Dept.~of Computer Sciences. The CHTC is supported by UW-Madison
and the Wisconsin Alumni Research Foundation, and is an active member
of the Open Science Grid, which is supported by the National Science
Foundation and the U.S. Department of Energy's Office of Science.  This work was partially supported by NSF Grants CCF-1148243 and CCF-1217058 and ONR award N00014-13-1-0129.

\bibliographystyle{IEEE}

\appendix
\section{Dual Subgradient Ascent}
\label{sec.dual-subgradient-ascent}
We note that ADMM is not the only method to decompose 
LP decoding problem. Here we present another decomposition method
using an (un-augmented) Lagrangian. This method is not 
as efficient as the ADMM decomposition. We hope to share 
this algorithm for readers interested in developing other
decomposition methods.

First we construct an un-augmented Lagrangian
\begin{align*}
L_0(\bfx, \bfz ,\bflambda) := \ & \bfgamma^T \bfx + \sum_{j \in
  \mathcal{J}} \bflambda_j^T ( \bfP_j \bfx - \bfz_j )
\end{align*}
the dual subgradient ascent method consists of the iterations:
\begin{align*}
\bfx^{k+1} & := \argmin_{\bfx \in \mathcal{X}} L_0(\bfx,\bfz^k, \bflambda^k) \\
\bfz^{k+1} & := \argmin_{\bfz \in \mathcal{Z}}  L_0(\bfx^{k}, \bfz, \bflambda^k ) \\
\bflambda_j^{k+1} & := \bflambda_j^k + \mu  \left( \bfP_j \bfx^{k+1} -\bfz_j^{k+1}  \right).
\end{align*}
Note here that the $\bfx$ and $\bfz$ updates are computed with respect to
the $k$ iterates of the other variables, and can be done completely in parallel.

The $\bfx$-update corresponds to solving the very simple LP:
\begin{equation*}
\begin{array}{ll}
  \mbox{minimize} & \left(\mbf{\gamma} +   \sum_{j \in
  \mathcal{J}}\bfP_j^T
\bflambda_j^k \right)^T \bfx \\
\mbox{subject to} & \bfx \in [0,1]^N.
\end{array}
\end{equation*}

This results in the assignment:
\begin{align*}
\bfx^{k+1} & = \theta \left(-\mbf{\gamma} -   \sum_{j \in
  \mathcal{J}}\bfP_j^T
\bflambda_j^k \right)
\end{align*}
where 
\[
\theta(t) = \begin{cases} 1 & t>0 \\
0 & t\leq 0\end{cases}
\]
is the Heaviside function.

For the $\bfz$-update, we have to solve the following LP for each $j\in\mathcal{J}$:
\begin{equation}\label{eq:lp-ppd}
\begin{array}{ll}
  \mbox{maximize} & {\bflambda_j^k}^T \bfz_j \\
\mbox{subject to} & \bfz_j \in \PP_d.
\end{array}
\end{equation}
Maximizing a linear function over the parity polytope can be performed
in linear time. First, note that the optimal solution necessarily
occurs at a vertex, which is a binary vector with an even hamming weight. Let $r$ be the number of positive components in the
cost vector $\bflambda_j^k$.  If $r$ is even, the vector $\bfv \in
\PP_d$ which is equal to $1$ where $\bflambda_j^k$ is positive and
zero elsewhere is a solution of~\eq{lp-ppd}, as making any additional
components nonzero decreases the cost as does making any of the
components equal to $1$ smaller.  If $r$ is odd, we only need to
compare the cost of the vector equal to $1$ in the $r-1$ largest
components and zero elsewhere to the cost of the vector equal to $1$
in the $r+1$ largest components and equal to zero elsewhere.  

The procedure to solve~\eq{lp-ppd} is summarized in
Algorithm~\ref{Algorithm:lp-ppd}.  Note that finding the smallest
positive element and largest nonnegative element can be done in linear
time.  Hence, the complexity of Algorithm~\ref{Algorithm:lp-ppd} is $O(d)$.

While this subgradient ascent method is quite simple, it is requires
significantly more iterations than the ADMM method, and thus we did not
pursue this any further.

\begin{algorithm}
\caption{Given a binary $d$-dimensional vector $\mbf{c}$, maximize
  $\mbf{c}^T\bfz$ subject to $\bfz\in\PP_d$.}
\label{Algorithm:lp-ppd}
\begin{algorithmic}[1]

\STATE Let $r$ be the number of positive elements in $\mbf{c}$.
\IF{$r$ is even}
\STATE {\bf Return} $\bfz^*$ where $z_i^*=1$ if $c_i>0$ and $z_i^*=0$ otherwise.
\ELSE
\STATE Find index $i_p$ of the smallest positive element of $\mathbf{c}$.
\STATE Find index $i_n$ of the largest non-positive element of
$\mathbf{c}$.
\IF{$c_{i_p}>c_{i_n}$}
\STATE {\bf Return}  $\bfz^*$ where $\bfz_i^*=1$ if $c_i>0$, $z_{i_n}^*=1$,
and $z_i^*=0$ otherwise.
\ELSE
\STATE  {\bf Return}  $\bfz^*$ where $\bfz_i^*=1$ if $c_i>0$ and $i \neq i_p$, $z_{i_p}^*=0$,
and $z_i^*=0$ for all other $i$.
\ENDIF
\ENDIF
\end{algorithmic}
\end{algorithm}

\end{document}